\documentclass[twocolumn,twoside]{IEEEtran}

\ifCLASSINFOpdf

\else

\fi

\ifCLASSOPTIONcompsoc
\usepackage[caption=false, font=normalsize, labelfont=sf, textfont=sf]{subfig}
\else
\usepackage[caption=false, font=normalsize]{subfig}
\fi
\usepackage{lipsum}%
\usepackage[dvipsnames]{xcolor}

\usepackage{balance}
\usepackage{multicol}   
\usepackage{cite}
\usepackage{gensymb}
\usepackage{multirow}
\usepackage{graphics}
\usepackage{epsfig}
\usepackage{graphicx}
\usepackage{epstopdf}
\usepackage{textcomp}
\usepackage{amsmath}
\usepackage{mathtools}
\usepackage{filecontents}
\usepackage{lipsum,color}
\usepackage{amssymb}
\usepackage{float}

\usepackage{amsthm}

\usepackage{times} 
\usepackage{amsthm}  

\usepackage{amsfonts}

\newtheorem{proposition}{Proposition}

\theoremstyle{break}

\begin{document}
\title{A Unified Framework for UAV-Based Free-Space Quantum Links: Beam Shaping and Adaptive Field-of-View Control}

\author{Mohammad~Taghi~Dabiri, ~Mazen~Hasna,~{\it Senior Member,~IEEE}, \\ ~Saif~Al-Kuwari,~{\it Senior Member,~IEEE}, ~and~Khalid~Qaraqe,~{\it Senior Member,~IEEE}
	
	\thanks{Mohammad Taghi Dabiri and Mazen Hasna are with the Department of Electrical Engineering, Qatar University, Doha, Qatar (e-mail: m.dabiri@qu.edu.qa; hasna@qu.edu.qa).}
	
	\thanks{Saif Al-Kuwari is with the Qatar Center for Quantum Computing, College of Science and Engineering, Hamad Bin Khalifa University, Doha, Qatar. email: (smalkuwari@hbku.edu.qa).}
	
	\thanks{Khalid A. Qaraqe is with the College of Science and Engineering, Hamad Bin Khalifa University, Doha, Qatar, and also with the Department of Electrical Engineering, Texas A\&M University at Qatar, Doha, Qatar (e-mail: kqaraqe@hbku.edu.qa).}
	
	\thanks{This publication was made possible by NPRP14C-0909-210008 from the Qatar Research, Development and Innovation (QRDI) Fund (a member of Qatar Foundation), Texas A\&M University at Qatar, and Hamad Bin Khalifa University, which supported this publication. } 
	
}

\maketitle
\begin{abstract}
This paper develops a comprehensive analytical framework for modeling and performance evaluation of unmanned aerial vehicles (UAVs)-to-ground quantum communication links, incorporating key physical impairments such as beam divergence, pointing errors at both transmitter and receiver, atmospheric attenuation, turbulence-induced fading, narrow field-of-view (FoV) filtering, and background photon noise. To overcome the limitations of conventional wide-beam assumptions, we introduce a grid-based approximation for photon capture probability that remains accurate under tightly focused beams. Analytical expressions are derived for the quantum key generation rate and quantum bit error rate (QBER), enabling fast and reliable system-level evaluation. Our results reveal that secure quantum key distribution (QKD) over UAV-based free-space optical (FSO) links requires beam waists below 10\,cm and sub-milliradian tracking precision to achieve Mbps-level key rates and QBER below $10^{-3}$. Additionally, we highlight the critical role of receiver FoV in balancing background noise rejection and misalignment tolerance, and propose adaptive FoV tuning strategies under varying illumination and alignment conditions. The proposed framework provides a tractable and accurate tool for the design, optimization, and deployment of next-generation airborne quantum communication systems.
\end{abstract}


\IEEEpeerreviewmaketitle


\section{Introduction}

The rapid growth of digital communication networks has made secure data exchange a cornerstone of modern society. However, the emergence of quantum computing threatens the foundations of classical cryptographic schemes such as RSA and ECC, which rely on computational hardness. In particular, the famous Shor algorithm can efficiently break these systems when running on a fault-tolerant quantum computer, making today's security systems vulnerable in a post-quantum era \cite{sood2024cryptography,tom2023quantum}.
To counteract this threat, quantum key distribution (QKD) has been proposed as a physically secure alternative, leveraging the laws of quantum mechanics to establish provably secure keys \cite{nadlinger2022experimental}. Unlike classical encryption, QKD ensures that any eavesdropping attempt introduces detectable disturbances, thus enabling unconditional security at the physical layer.

While fiber-based QKD has demonstrated high security over controlled links, its distance limitation and high deployment cost restrict its scalability, particularly in dynamic or large-scale scenarios \cite{kong2024secret}. Free-space QKD (FS-QKD), on the other hand, offers infrastructure-free deployment and greater flexibility \cite{kong2022unmanned}. However, implementing FS-QKD remains challenging due to its reliance on precise line-of-sight (LoS) alignment and sensitivity to atmospheric and geometric impairments. Among airborne platforms, unmanned aerial vehicles (UAVs) have gained significant interest for QKD applications, offering mobility, rapid deployment, and favorable elevation angles for establishing LoS quantum links \cite{quintana2019low,trinh2024quantum}.

Despite their potential, UAV-based QKD systems face stringent challenges, including beam misalignment, atmospheric turbulence, and background photon noise. Unlike classical free-space optical (FSO) links, quantum systems rely on single-photon transmissions, which inherently limit the key generation rate, primarily due to geometric losses. Only a small portion of photons reach the receiver aperture, making sub-milliradian pointing accuracy and ultra-narrow field-of-view (FoV) essential to minimize alignment losses. These constraints are critical to achieving high quantum key rates while keeping the quantum bit error rate (QBER) low \cite{trinh2025towards}.
To address these challenges, we develop an accurate and tractable model for UAV-to-ground quantum channels that captures key physical impairments and quantifies their joint impact on QKD performance.

\subsection{Literature Review}
While extensive research has been conducted on fiber-based QKD, the modeling and performance analysis of FSO quantum communication systems—particularly over long-distance links—remain relatively limited \cite{ralegankar2021quantum}, potentially restricting the practical utility of FSO-based QKD. Unlike optical fibers, which offer stable and well-characterized propagation environments, free-space quantum links introduce significant challenges, such as severe geometric loss, atmospheric turbulence, and background noise, especially in dynamic scenarios involving airborne platforms.

Several recent studies have explored the potential of UAVs in quantum communication, particularly in free-space QKD scenarios. A notable survey is provided in~\cite{kumar2021survey}, where various quantum-enabled UAV applications, including QKD, entanglement distribution, and sensing—are systematically reviewed. \cite{gao2023qkd} demonstrated a lightweight software-based QKD for UAV-to-ground links with practical feasibility. Kondamuri et al.\ presented analytical error and capacity models for coherent-state quantum modulation under weak turbulence \cite{kondamuri2024quantum}, while Waghmare et al.\ revealed a trade-off between QBER and signal photon count in turbulent free-space channels \cite{waghmare2023performance}.  
However, these studies largely overlook the critical impact of pointing errors and geometrical loss, which are fundamental limitations in practical long-distance free-space quantum communication systems.


Building on these challenges, Nguyen et al.\ proposed a blind reconciliation scheme using protograph LDPC codes for FSO-based satellite QKD, with pointing error modeled via a plane-wave approximation \cite{nguyen2024blind}. Trinh et al.\ introduced rooftop-mounted optical reconfigurable intelligent surfaces (ORISs) to improve HAP-to-UAV QKD by adaptively mitigating pointing errors and turbulence effects \cite{trinh2025optical}. Vazquez-Castro and Samandarov showed that quantum detection outperforms classical methods for binary modulations in space channels across orbital regimes \cite{vazquez2023quantum}, while Chakraborty et al.\ developed a hybrid noise model for more accurate secret key rate (SKR) estimation in satellite-based QKD \cite{chakraborty2025hybrid}. Despite these advances, pointing error remains the dominant limitation in long-distance quantum links, where even slight misalignments can reduce secret key rates to just a few bits per second—or even lower—due to the single-photon nature of quantum signals \cite{dabiri2025impact}.

Although the aforementioned studies focus on long-distance satellite-based links, this class of work typically targets short-range atmospheric channels in the kilometer scale, such as UAV-based scenarios. For instance, \cite{alshaer2021reliability,alshaer2022performance} analyze the impact of pointing errors and UAV dynamics in entanglement- and QKD-based FSO systems, proposing optimized configurations under realistic turbulence conditions. Kong \cite{kong2024uav} proposed using UAVs as mobile couriers to deliver QKD-generated keys to ground devices, addressing a key deficiency in resource-limited nodes. Xue et al.\ \cite{xue2021airborne} reviewed airborne QKD progress, emphasizing its flexibility for extending terrestrial and satellite networks. Al-Badarneh et al.\ \cite{al2025unified} developed a unified MGF-based framework to evaluate UAV-based quantum links under turbulence and pointing errors using Helstrom detection.

All existing channel models in both satellite and short-range UAV-based quantum communication studies rely on the pointing error formulation presented in \cite{farid2007outage}, originally developed for conventional FSO links. This model assumes a wide-beam regime where the beamwidth at the receiver $w_z$ is much larger than the aperture radius $r_a$. However, while this condition is valid for long-distance quantum satellite links, its applicability to short-range terrestrial or UAV-based quantum systems remains uncertain.

This raises an important question: to what extent are classical FSO models valid for quantum communication links? This question arises from the fundamental differences between classical and quantum systems.
In classical systems, the channel transmissivity can be as low as $10^{-4}$ or even less, yet the high photon flux ensures that enough power is received. The performance metric is typically outage probability, and hence the beam width is optimized to be large (on the order of meters) to tolerate pointing errors and maintain a reliable link \cite{dabiri2019tractable,dabiri2019optimal}.
In contrast, quantum systems operate with single-photon-level transmissions. A transmissivity of $10^{-4}$ implies that only one photon is received every $10^4$ time slots. Moreover, quantum information is often carried by entangled photons, and only reveals its content upon measurement—suitable for specialized applications like QKD or quantum sensing. After detection, key reconciliation over a public channel is required, further shifting the focus from outage probability to SKR.

Therefore, unlike classical FSO, the dominant metric in quantum systems is the average channel transmissivity, which must be increased—ideally to the order of $10^{-2}$—to achieve usable SKRs. This can only be accomplished by using tighter tracking systems and narrower beam widths, typically around 10 cm. However, under such narrow-beam conditions, the wide-beam approximation in \cite{farid2007outage} becomes invalid. To the best of our knowledge, none of the existing analytical models has addressed this inconsistency.
Furthermore, achieving high SKR and maintaining the required QBER threshold demands a holistic system-level design that jointly optimizes parameters such as photon generation rate, background noise, and receiver field-of-view (FoV) under varying ambient and weather conditions. This coupled design problem remains an open gap in the literature—one that this paper aims to address.

\subsection{Contributions}

This paper addresses a critical modeling gap in UAV-based free-space quantum communication systems, where the widely adopted geometric loss model fails under the narrow-beam conditions essential for high SKR performance. Recognizing the fundamental differences between classical and quantum FSO links—particularly in terms of channel transmissivity requirements and system design priorities—we propose a new analysis framework that revises the pointing error model beyond the conventional far-field assumption. Furthermore, we perform a joint system-level optimization over key design parameters, including beam width, FoV, photon emission rate, and background noise, to identify operating regimes that maximize SKR while satisfying QBER constraints under realistic turbulence and ambient light conditions.

\begin{itemize}
	\item \textbf{Modeling validity analysis}: We rigorously examine the validity limits of the conventional pointing error model~\cite{farid2007outage} under narrow-beam conditions (e.g., $w_z \approx 10$~cm), which are critical for achieving high SKR in short-range quantum links.
	
	\item \textbf{Refined channel modeling}: We derive an improved pointing error formulation that captures beam divergence and geometric alignment effects without relying on the far-field approximation.
	
	\item \textbf{Performance metrics redefinition}: We shift the focus from outage probability (used in classical FSO) to average channel transmissivity as the key metric for evaluating quantum link performance.
	
	\item \textbf{System-level co-design}: We develop a coupled optimization framework that jointly tunes beam width, FoV, background noise acceptance, and photon emission rate to maximize SKR under QKD-specific constraints.
	
	\item \textbf{Design guidelines}: We provide practical insights for configuring UAV-based QKD links under varying turbulence and background light conditions, including identification of optimal FoV settings to balance alignment tolerance and noise rejection.
\end{itemize}

Although the core contributions of this work are tailored toward the design of UAV-to-ground quantum links, owing to their potential in enabling secure LoS connections between distributed users, the modeling and optimization framework developed herein is broadly applicable to other short-range FS-QKD configurations involving dynamic or stationary platforms.

\begin{figure*}
	\centering
	\subfloat[] {\includegraphics[width=2.2 in]{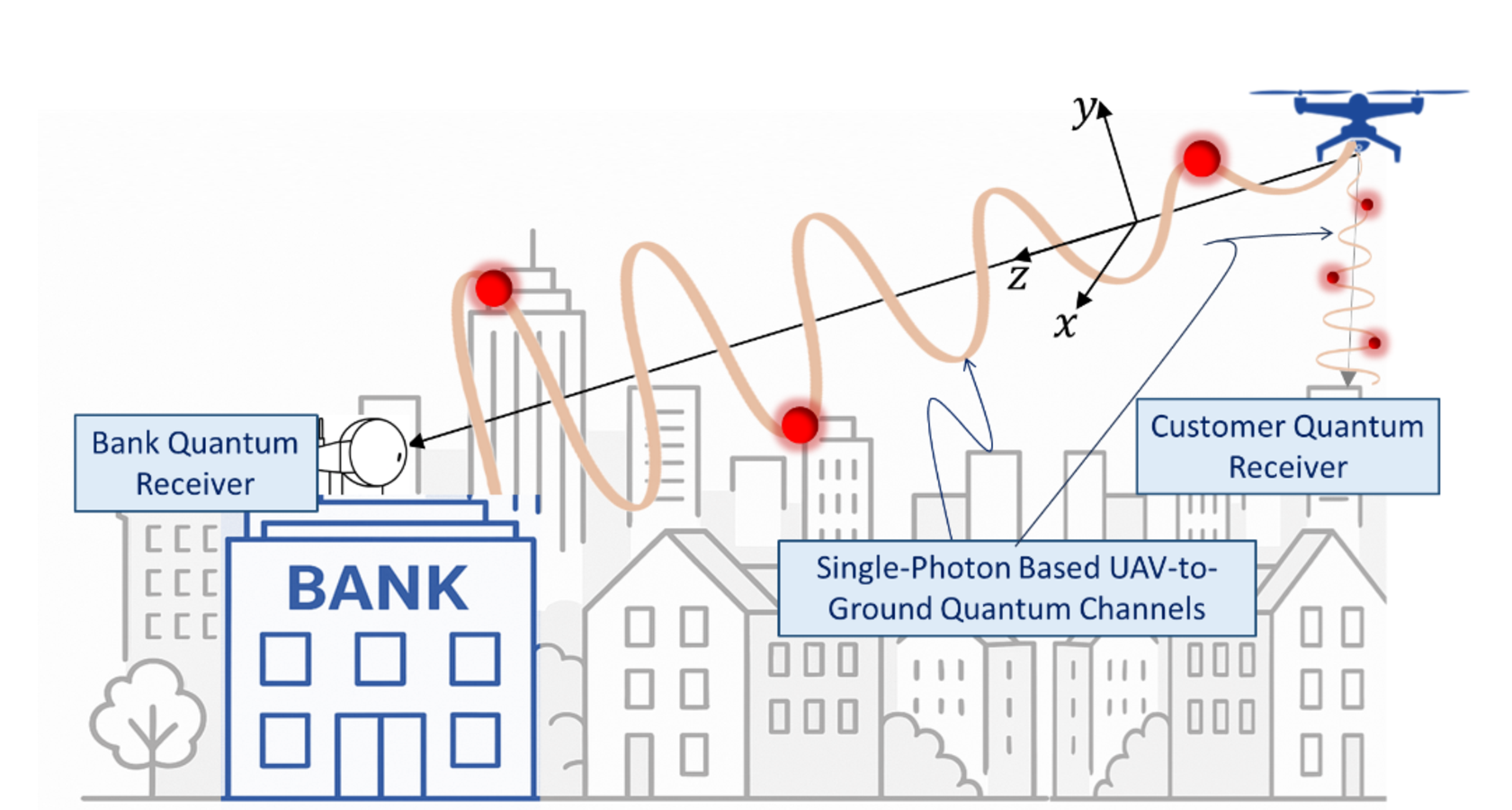}
		\label{ch1}
	}
	\hfill
	\subfloat[] {\includegraphics[width=3 in]{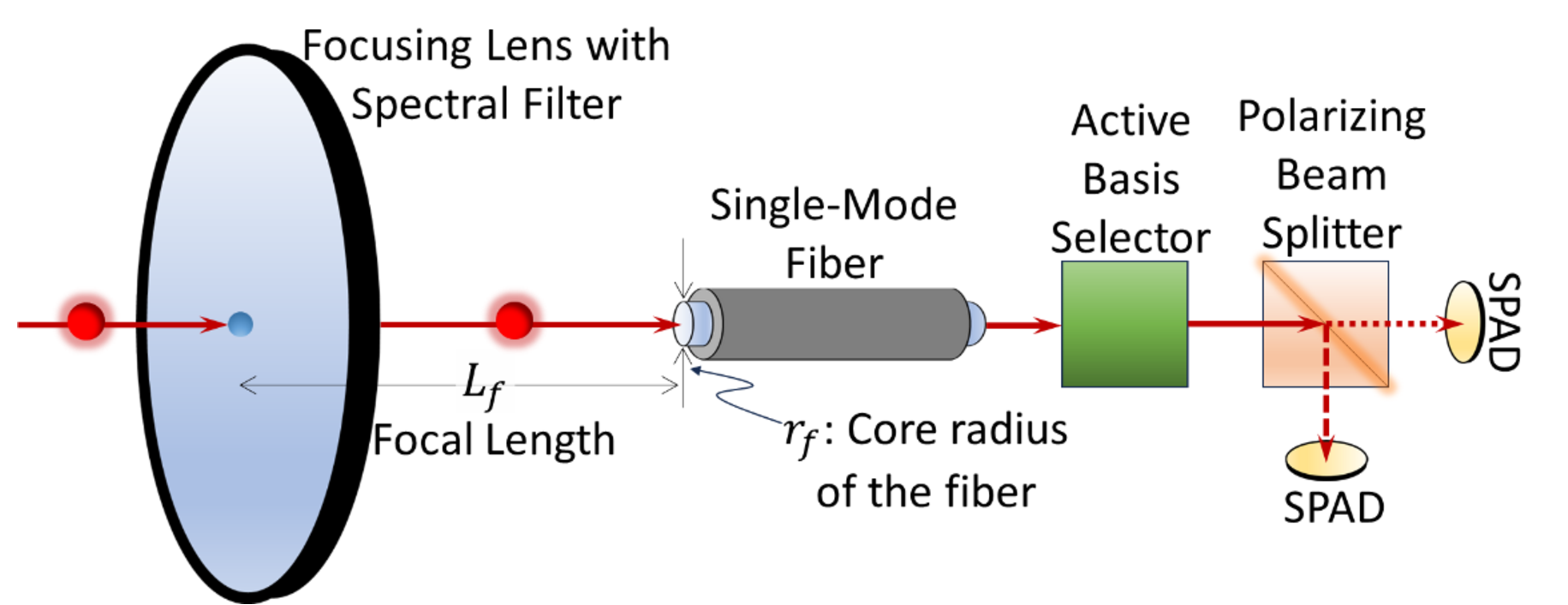}
		\label{ch2}
	}
    \hfill
    \subfloat[] {\includegraphics[width=1.7 in]{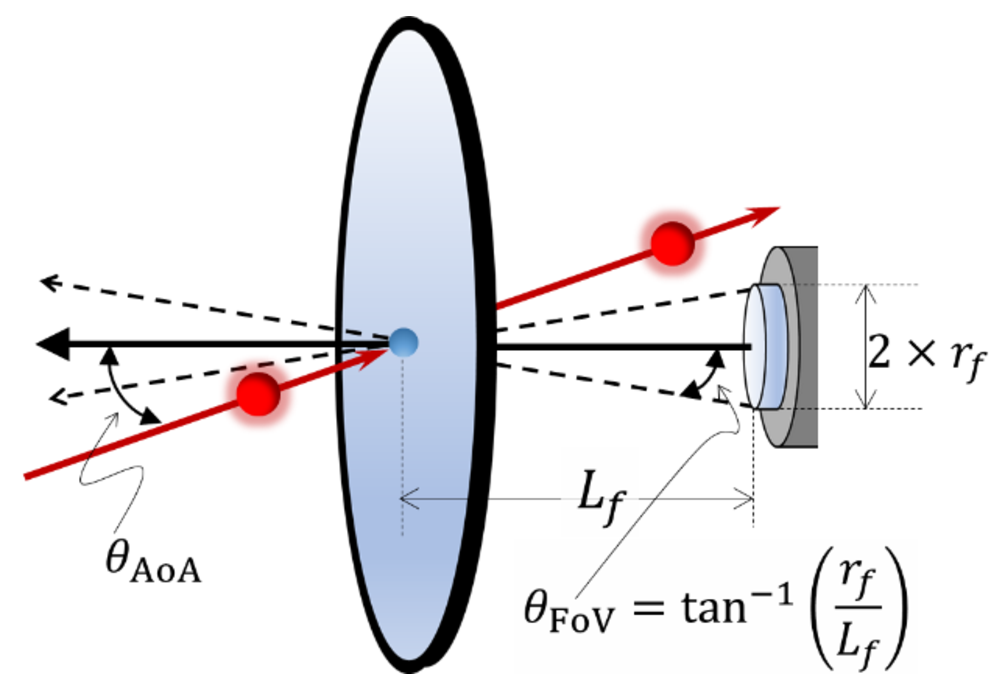}
	\label{ch3}
    }
	\caption{(a) Conceptual illustration of a UAV-to-ground quantum communication scenario, where a UAV acts as a quantum transmitter to establish a shared secret key between a bank and a customer. The UAV facilitates line-of-sight optical links with both ground nodes, enabling secure QKD in a mobile and flexible setting. (b) Simplified block diagram of the photon receiver. The incoming photon or quantum bit (qubit) is collected by a converging lens with radius $r_{\text{lens}}$ and focused at its focal length $L_f$ onto the core of a single-mode fiber with radius $r_f$. A spectral filter positioned between the lens and the fiber removes unwanted wavelengths. After fiber coupling, the photon passes through a polarizing beam splitter (PBS) and an active basis selector, which together determine the measurement basis. The photon is then detected by a single-photon avalanche diode (SPAD), enabling the projection of the qubit onto a specific polarization state for quantum information extraction.
    (c) Illustration of the receiver's angular FoV defined by the single-mode fiber core radius $r_f$ and the lens focal length $L_f$. The angular deviation of the incoming beam, denoted by $\boldsymbol{\theta}_{\text{AoA}}$, must remain within $\theta_{\text{FoV}}$ for the qubit to be accepted.}
	
	\label{ch}
\end{figure*}

\begin{table}[t]
	\centering
	\caption{Summary of System Notations}
	\begin{tabular}{ll}
		\hline
		\textbf{Notation} & \textbf{Description} \\
		\hline
		$n_t$ & Number of photons in a transmitted pulse \\
		$\mu_t$ & Average photon number per transmitted pulse \\
		$T_{qs}$ & Quantum slot duration \\
		$R_q$ & Quantum transmission rate ($R_q = 1/T_{qs}$) \\
		$x, y, z$ & Spatial coordinates (z is the propagation axis) \\
		$w_0$ & Initial beam waist at transmitter \\
		$w_z$ & Beam radius at distance $z$ \\
		$\lambda$ & Operating wavelength \\
		$L_z$ & UAV-to-ground distance \\
		$\boldsymbol{\theta}_e$ & Transmitter FSM pointing error (2D) \\
		$\sigma_{\theta_e}^2$ & Variance of transmitter FSM-based tracking error \\
		$r_d$ & Beam lateral displacement at receiver plane \\
		$\sigma_{r_d}^2$ & Variance of beam displacement due to pointing error \\
		$r_a$ & Receiver aperture radius \\
		$\mu_p$ & Photon capture probability within aperture \\
		$N_g$ & Number of grid segments in aperture approximation \\
		$\eta_{\text{atm}}$ & Atmospheric transmittance \\
		$\alpha_a$ & Atmospheric attenuation coefficient \\
		$\eta_{\text{turb}}$ & Turbulence-induced transmittance \\
		$\alpha, \beta$ & Gamma-Gamma turbulence parameters \\
		$\sigma_R^2$ & Rytov variance (index of turbulence strength) \\
		$n_b$ & Number of background photons per time slot \\
		$\mu_b$ & Mean background photon count \\
		$B_{\lambda}$ & Spectral radiance of background light \\
		$A_r$ & Effective receiver aperture area \\
		$\Delta\lambda$ & Spectral filter bandwidth \\
		$\hbar, c$ & Planck’s constant and speed of light \\
		$r_f$ & Radius of single-mode fiber core \\
		$L_f$ & Focal length of the lens \\
		$\theta_{\text{FOV}}$ & Angular field of view \\
		$\Omega_{\text{FOV}}$ & Solid-angle field of view \\
		$\boldsymbol{\theta}_{\text{AoA}}$ & Angle of arrival at receiver \\
		$\sigma_{\text{AoA}}^2$ & Variance of angle of arrival \\
		$\mu_{\text{FOV}}$ & Visibility indicator (0 or 1) \\
		$\mu_d$ & Detector efficiency \\
		$n_q$ & Number of detected photons per quantum time slot \\
		$\mu_q$ & Mean detected photon count per slot \\
		\hline
	\end{tabular}
	\label{tab:notations}
\end{table}

\section{System Model}
We consider a UAV-to-ground quantum communication system in which a mobile UAV acts as the transmitter, while a fixed ground station serves as the receiver. Fig. \ref{ch1} illustrates this scenario. The UAV is equipped with a quantum photon source that emits single-photon pulses encoded in polarization states. These photons propagate through a FSO channel, which is subject to several impairments such as beam wandering, atmospheric turbulence, pointing errors, and absorption. The receiver on the ground includes optical focusing and filtering components, polarization-based separation, and single-photon detection units. The primary goal is to model the quantum channel between the UAV and the ground station, taking into account the dynamic behavior of the UAV and the statistical properties of the atmospheric channel.

\subsection{Transmitter}
In the considered UAV-to-ground quantum communication setup, the UAV is equipped with an onboard transmitter capable of emitting single-photon pulses. These photons are typically generated using a weak coherent pulse (WCP) source, where the number of photons per pulse follows a Poisson distribution. The probability that a pulse contains exactly $k$ photons is given by \cite{foletto2022security}
\begin{align} \label{df1}
	\mathbb{P}(n_t = k) = \frac{e^{-\mu_t} \mu_t^k}{k!},
\end{align}
where $\mu_t$ denotes the average photons per pulse. 
The selection of $\mu_t$ plays a critical role in the overall system performance.\footnote{In many practical scenarios, the instantaneous transmissivity of the quantum channel can fall below $10^{-4}$, implying that, on average, fewer than one photon per $10^4$ quantum time slots may be successfully received when $\mu_t = 1$. This drastically reduces the effective quantum bit rate. An intuitive approach to mitigate this limitation is to increase $\mu_t$; however, excessively high values of $\mu_t$ compromise the security of quantum protocols by increasing the probability of multi-photon emissions. This trade-off between photon rate and security will be further analyzed in the following sections.}
The quantum source operates in discrete time slots of duration $T_{qs}$, referred to as the quantum slot duration. This parameter is inversely related to the photon transmission rate $R_q$, such that $T_{qs} = 1/R_q$.

To describe the spatial configuration, we define a coordinate system in which the direction of photon propagation from the UAV to the ground station lies along the $z$-axis. The apertures of both the transmitter and the receiver are located in the $(x-y)$-plane, which is orthogonal to the transmission axis. 
Photons emitted from the UAV transmitter are assumed to propagate in a fundamental Gaussian mode, characterized by an initial beam waist radius $w_0$, which governs the beam's divergence during free-space propagation. This mode is a common approximation for sources operating in the $\text{TEM}_{00}$ configuration \cite{pampaloni2004gaussian}. 
Unlike classical optics, where intensity represents power distribution, in quantum communication, the spatial beam profile defines the probability density of a photon's presence. Specifically, the normalized probability of detecting a single photon at a transverse position $(x, y)$ and distance $z$ is given by
\begin{align} \label{xc1}
	|\psi(x, y, z)|^2 = \frac{2}{\pi w_z^2} \exp\left(-\frac{2(x^2 + y^2)}{w_z^2} \right),
\end{align}
where $w_z$ denotes the beam radius at distance $z$, defined as the point where the probability density falls to $1/e^2$ of its maximum. At the UAV-to-ground separation distance $L_z$, this radius becomes \cite{ghassemlooy2019optical}
\begin{align}
	w_{L_z} = w_0 \sqrt{1 + \left( \frac{\lambda L_z}{\pi w_0^2} \right)^2 },
\end{align}
where $\lambda$ is the operating wavelength. This probabilistic description reflects the intrinsic quantum uncertainty of the photon's transverse position prior to measurement. The photon remains delocalized until detection occurs, and its spatial behavior is governed by the wavefunction's modulus squared.
In contrast, classical FSO systems involve coherent beams with a large number of photons per pulse, enabling the use of deterministic optical intensity. In such systems, the Gaussian beam profile corresponds directly to measurable power distribution, whereas in the quantum regime, it characterizes the likelihood of individual photon detection.

Moreover, due to the motion and vibration of the UAV, as well as environmental disturbances, the alignment between the transmitter and the receiver is imperfect. To mitigate this, we assume the use of a fast-steering mirror (FSM) as a lightweight, responsive tracking system onboard the UAV. However, residual pointing errors persist and are modeled as a two-dimensional random angular deviation $\boldsymbol{\theta}_e = (\theta_{ex}, \theta_{ey})$, with each component following a zero-mean Gaussian distribution \cite{dabiri2025novel}:
\begin{align} \label{xc2}
	\theta_{ex}, \theta_{ey} \sim \mathcal{N}(0, \sigma_{\theta_e}^2).
\end{align}

\subsection{Photon Propagation}
As photons travel from the UAV transmitter to the ground-based receiver, they experience various environmental impairments as a result of atmospheric effects. Among these, potential changes in polarization due to scattering and refraction are often considered. However, such effects are predominantly induced by off-axis scattering events or non-line-of-sight (NLOS) propagation paths \cite{zhang2012characteristics}. Since the receiver only detects photons that remain on or near the direct line-of-sight trajectory, and since quantum communication systems typically employ narrow beams and tight acquisition mechanisms, the depolarization effects along the direct path can be considered negligible \cite{zhang2012characteristics}. Therefore, we assume that the polarization state of the photons is preserved during propagation through the atmosphere.

Generally, two primary mechanisms contribute to photon attenuation along the channel: atmospheric absorption and scattering, and random fluctuations in the channel transmissivity due to turbulence.
The deterministic attenuation caused by molecular absorption and scattering is modeled using the Beer–Lambert law. The atmospheric transmittance $\eta_{\text{atm}}$ over a propagation distance $L_z$ is given by \cite{ghassemlooy2019optical}
\begin{align} \label{at1}
	\eta_{\text{atm}} = \exp(-\alpha_a L_z),
\end{align}
where $\alpha_a$ is the atmospheric attenuation coefficient, which depends on wavelength, altitude, and visibility conditions.

In addition to deterministic loss, atmospheric turbulence causes random fluctuations in the received signal intensity as a result of refractive index variations along the propagation path. These scintillation effects are well-modeled by the Gamma-Gamma distribution, particularly under moderate-to-strong turbulence. The turbulence-induced transmissivity $\eta_{\text{turb}}$ is modeled as a continuous random variable with the following probability density function (PDF) \cite{andrews2005laser}:
\begin{align} \label{at2}
	f_{\eta_{\text{turb}}}(\eta) = \frac{2(\alpha \beta)^{\frac{\alpha + \beta}{2}}}{\Gamma(\alpha)\Gamma(\beta)} \eta^{\frac{\alpha + \beta}{2} - 1} K_{\alpha - \beta} \left(2\sqrt{\alpha \beta \eta} \right),
\end{align}
where $\alpha$ and $\beta$ represent the small- and large-scale turbulence parameters, respectively, both of which are functions of the Rytov variance $\sigma_R^2$.

\subsection{Receiver}
After traversing the atmospheric channel and experiencing various impairments, such as beam divergence, turbulence-induced fading, and tracking misalignment, a qubit encoded in the polarization state of a single photon has a very small probability of reaching the receiver aperture. As illustrated in Fig.~\ref{ch2}, the quantum receiver consists of a simplified photonic structure that captures and processes the incoming qubit.
The aperture comprises a converging lens with radius $r_{\text{lens}}$ that collects the incoming photon and focuses it onto the end-face of a single-mode fiber (SMF). The fiber core has a radius $r_f$ and is precisely positioned at the focal plane of the lens, whose focal length is denoted by $L_f$. An optical spectral filter is placed immediately after the lens and before the fiber to block background photons outside the operational wavelength band, thereby improving signal-to-noise performance.
After being spectrally filtered and coupled into the fiber, the qubit is directed toward a set of optical and detection components. A polarizing beam splitter (PBS), combined with an active basis selector, defines the measurement basis. The qubit is then projected onto a specific polarization state and detected by a single-photon avalanche diode (SPAD), enabling quantum bit discrimination with high temporal precision.

One of the primary challenges in UAV-to-ground quantum communication is the presence of background photon noise, especially during daylight operations. This issue arises due to two main factors. First, the receiver aperture is oriented skyward, inherently exposing it to higher levels of ambient light from the sun and atmospheric scattering. Second, the quantum channel is highly lossy, meaning that in most quantum time slots, the intended signal photon does not reach the detector. As a result, any ambient photon arriving within the receiver's FOV and matching the operational wavelength of the quantum system, even with an incorrect polarization, can be detected by one of the SPADs, thereby leading to a false quantum key bit detection.
The number of background photons incident during a single quantum time slot follows a Poisson distribution with mean $\mu_b$, given by \cite{kim2022photon}
\begin{align}
	\mathbb{P}(n_b = k) = \frac{e^{-\mu_b} \mu_b^k}{k!},
\end{align}
where $n_b$ denotes the number of background photons in a single time slot, and $\mu_b$ is the average background photon count per slot.
The mean value $\mu_b$ depends on several physical and optical parameters, including the receiver's FOV, the background spectral radiance, and the quantum slot duration $T_{qs}$. 
The mean number of background photons incident during a quantum slot is given by
\begin{align}
	\mu_b = \frac{B_{\lambda} \cdot A_r \cdot \Omega_{\text{FOV}} \cdot \Delta\lambda \cdot T_{qs} \cdot \lambda}{\hbar c},
\end{align}
where $B_{\lambda}$ is the spectral radiance of the background light (in W/m$^2$/sr/nm), $A_r$ is the effective receiver aperture area, $\Omega_{\text{FOV}}$ is the receiver field of view, $\Delta\lambda$ is the spectral filter bandwidth, $T_{qs}$ is the quantum slot duration, $\lambda$ is the operating wavelength, and $\hbar$ and $c$ are Planck's constant and the speed of light, respectively.

The effective receiver FOV plays a key role in determining the number of background photons that can enter the detection path. As shown in Fig.~\ref{ch3}, the angular FOV, denoted by $\theta_{\text{FOV}}$, can be geometrically defined as
\begin{align}
	\theta_{\text{FOV}} = \tan^{-1}\left(\frac{r_f}{L_f}\right),
\end{align}
where $r_f$ is the core radius of the single-mode fiber and $L_f$ is the focal length of the converging lens. Assuming a circular symmetric conical field of view with half-angle $\theta_{\text{FOV}}$, the solid-angle field of view is given by
\begin{align}
	\Omega_{\text{FOV}} = 2\pi \left(1 - \cos(\theta_{\text{FOV}}) \right).
\end{align}
To minimize the impact of background noise, it is critical to keep $\theta_{\text{FOV}}$ significantly smaller than those used in conventional FSO systems. This is particularly important in UAV-based quantum links, where the receiver aperture typically points toward the sky and thus collects more ambient light. Reducing $\theta_{\text{FOV}}$ limits the background photon influx, but also increases the system's sensitivity to angular misalignment.

As depicted in Fig.~\ref{ch3}, the incident angle of the incoming beam with respect to the receiver's optical axis, denoted by $\boldsymbol{\theta}_{\text{AoA}} = (\theta_{\text{AoA},x}, \theta_{\text{AoA},y})$, becomes a critical parameter. This angular deviation is primarily due to imperfections in the stabilization system that maintains the receiver aperture aligned in the $x$–$y$ plane. We model this angular deviation as a bivariate Gaussian distribution to capture the statistical misalignment induced by mechanical stabilization errors \cite{dabiri2018channel}:
\begin{align} \label{aoa1}
	(\theta_{\text{AoA},x}, \theta_{\text{AoA},y}) \sim \mathcal{N}(0, \sigma_{\text{AoA}}^2),
\end{align}
where $\sigma_{\text{AoA}}^2$ characterizes the variance of angular misalignment in each axis.
Unlike the fine angular correction performed by the transmitter-side FSM with tracking error $\boldsymbol{\theta}_e$, the receiver-side stabilization system generally has slower response and coarser resolution. 
To account for the directional filtering introduced by the narrow FOV, we define a binary visibility indicator $\mu_{\text{FOV}}$, which models whether the incoming qubit falls within the receiver's acceptance cone:  
\begin{align} \label{at3}
	\mu_{\text{FoV}} = 
	\begin{cases}
		1, & \text{if } \lVert \boldsymbol{\theta}_{\text{AoA}} \rVert \leq \theta_{\text{FoV}}, \\
		0, & \text{otherwise}.
	\end{cases}
\end{align}

Finally, the coupled photon is directed toward one of the SPAD detectors based on its polarization state, as determined by the active basis selection module.
The probability that a photon incident on the SPAD induces a detectable avalanche event is governed by the detector efficiency, denoted by $\mu_d<1$.
This parameter incorporates the effects of quantum efficiency, internal avalanche triggering probability, and circuit-level losses. For an ideal detector, $\mu_d = 1$, but in practice, values between 0.5 and 0.8 are common, depending on the wavelength and detector technology.
%

\section{Analytical Framework for Channel Modeling and Performance Evaluation}
In this section, we develop a unified analytical model for the UAV-to-ground quantum communication channel by integrating the effects of beam divergence, atmospheric absorption and scattering, turbulence-induced fading, pointing errors, receiver misalignment, and FoV filtering. While each of these phenomena has been characterized individually in the system model, their joint impact must be captured to accurately describe the channel's end-to-end behavior. 

Due to the transmitter-side FSM tracking error modeled in \eqref{xc2}, the outgoing photon beam may deviate from the intended propagation axis. This angular deviation, denoted by $\boldsymbol{\theta}_e = (\theta_{ex}, \theta_{ey})$, translates into a lateral displacement at the receiver plane. Specifically, at a transmission distance of $L_z$, the deviation in position can be expressed as $\boldsymbol{r}_d = (r_{dx}, r_{dy})$, where
\begin{align}
	r_{dx} = \theta_{ex} L_z, \quad r_{dy} = \theta_{ey} L_z.
\end{align}
Since each angular component in \eqref{xc2} follows a zero-mean Gaussian distribution with variance $\sigma_{\theta_e}^2$, the corresponding lateral displacement components $r_{dx} $ and $r_{dy} $ are also independent Gaussian variables with zero mean and variance $\sigma_{r_d}^2 = \sigma_{\theta_e}^2 L_z^2$.  This spatial shift displaces the center of the beam's transverse probability distribution, as modeled in \eqref{xc1}, leading to the misaligned photon profile:
\begin{align} \label{xc4}
	|\psi(x, y, z \mid \boldsymbol{r}_d)|^2 = \frac{2}{\pi w_z^2} \exp\left(-2\frac{(x - r_{dx})^2 + (y - r_{dy})^2}{w_z^2} \right).
\end{align}
Using \eqref{xc4}, the probability that a photon lands within the receiver aperture of radius $r_a$ is denoted by $\mu_p$. This is obtained by integrating the spatial probability density over the circular detection region:
\begin{align} \label{xc6}
	&\mu_p(\boldsymbol{r}_d) = \frac{2}{\pi w_z^2} \nonumber \\
	&\times\iint\limits_{x^2 + y^2 \leq r_a^2} \exp\left(-\frac{2[(x - r_{dx})^2 + (y - r_{dy})^2]}{w_z^2} \right) dx\,dy.
\end{align}

In classical FSO systems, it is common to assume that the beam size is significantly larger than the receiver aperture, i.e., $w_z \gg r_a$. Under this assumption, the variation of the Gaussian beam over the aperture area is negligible, and the integral in \eqref{xc6} simplifies to \cite{farid2007outage,dabiri2024modulating}:
\begin{align} \label{xc7}
	\mu_p(\boldsymbol{r}_d) \approx  \frac{2 r_a^2}{w_z^2} \exp\left(-\frac{2\|\boldsymbol{r}_d\|^2}{w_z^2} \right).
\end{align}
This approximation treats the Gaussian profile as nearly constant across the aperture, effectively pulling it outside the integral. However, as will be discussed in Section \ref{simulation}, this assumption does not hold in quantum communication scenarios. First, unlike FSO systems, where the primary concern is link outage, quantum channels are inherently high-loss due to their single-photon nature; thus, outage probability becomes a secondary metric rather than a defining performance criterion. Second, increasing $w_z$ to satisfy the $w_z \gg r_a$ condition leads to an extremely diffuse beam, which in turn drastically reduces the probability of photon detection at the receiver. Since quantum links transmit only one photon per pulse, large beam divergence can reduce the secret key rate to below one bit per second. Therefore, selecting a smaller $w_z$ is not only acceptable but also desirable for quantum key generation, even if it increases the sensitivity to alignment errors. This trade-off between photon concentration and misalignment robustness is a key distinction between classical and quantum free-space optical links.
Therefore, in scenarios where the assumption $w_z \gg r_a$ no longer holds, the approximation in \eqref{xc7} becomes invalid. In such cases, the full two-dimensional integral in \eqref{xc6} must be evaluated. 
To gain a better understanding of the limitations of the classical wide-beam approximation in \eqref{xc7}, Fig.~\ref{sm1} illustrates the photon capture probability $\mu_p(r_d)$ for two different spot sizes, $w_z = 5$\,cm and $w_z = 10$\,cm. These values are representative of the practical settings used in our simulation section, where it is shown that such beam sizes are more suitable for UAV-based quantum communication scenarios aiming to achieve higher key generation rates. As we demonstrate in Fig. \ref{sm1}, the classical approximation in \eqref{xc7} noticeably deviates from the numerically evaluated results based on the exact integral formulation in \eqref{xc6}, particularly under realistic quantum operating conditions. This highlights the inadequacy of conventional FSO assumptions in high-loss quantum regimes and emphasizes the need for accurate beam-profile modeling.

%
\begin{figure}
	\begin{center}
		\includegraphics[width=3.0 in]{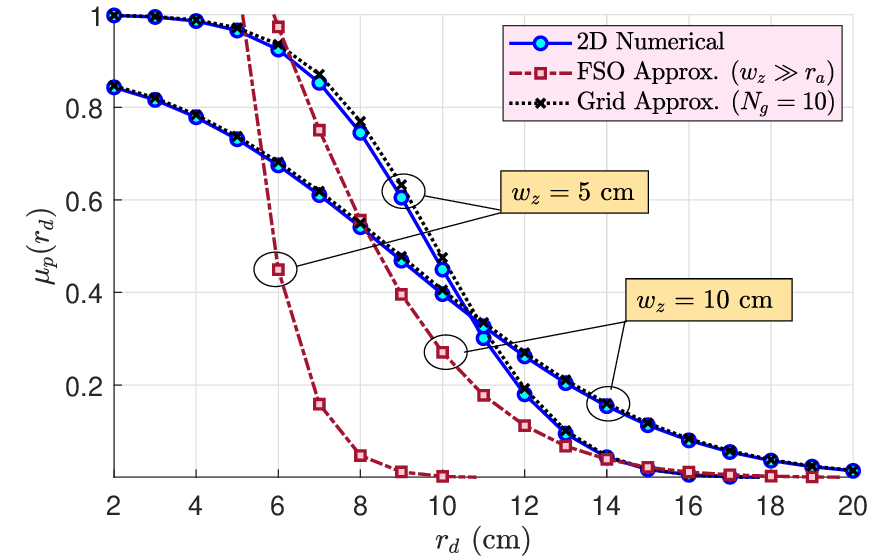}
		\caption{Photon capture probability $\mu_p(r_d)$ versus beam displacement $r_d$ for two different spot sizes: $w_z = 5$\,cm and $w_z = 10$\,cm. The figure compares the classical FSO approximation in \eqref{xc7} and the proposed grid-based model in Proposition~1 with the accurate numerical reference obtained from the two-dimensional integral in \eqref{xc6}.}
		\label{sm1}
	\end{center}
\end{figure}

\begin{proposition}
	The photon capture probability $\mu_p(r_d)$ given by the exact integral formulation in \eqref{xc6} can be efficiently approximated by using a grid-based discretization scheme. Specifically, dividing the receiver aperture interval $[-r_a, r_a]$ into $N_g$ equal-width segments of size $\Delta x = 2r_a / N_g$, the capture probability can be expressed as:
	\begin{align} \label{sb1}
		\mu_p(r_d) \approx \sum_{i=1}^{N_g} \mu_p(r_d, i).
	\end{align}
    where
	\begin{align} \label{sbb2}
		\mu_p(r_d, i) \approx  c_i \exp\left( -\frac{2(x_i - r_d)^2}{w_z^2} \right),
	\end{align}
	 $x_i$ is the center of the $i$-th segment, and
	\begin{align} 
		c_i = \frac{2\, \Delta x}{\sqrt{2\pi} w_z}  \cdot \text{erf}\left( \sqrt{ \frac{2}{w_z^2} } \sqrt{r_a^2 - x_i^2} \right).
	\end{align}
\end{proposition}
\begin{proof}
	Please refer to Appendix~\ref{AppA}.
\end{proof}
This approximation provides a computationally efficient method for evaluating $\mu_p(r_d)$ under realistic beam sizes where the classical approximation fails.
The validity of the grid-based approximation developed in Proposition~1 is further examined in Fig.~\ref{sm1}, which compares the approximate photon capture probability $\mu_p(r_d)$ with numerically evaluated values obtained from the exact two-dimensional integral expression in \eqref{xc6}. The figure shows that even with a coarse discretization using only $N_g = 10$ segments, the grid-based model provides an excellent match to the reference values across a wide range of beam displacements. This confirms that the proposed approximation achieves high accuracy while significantly reducing computational complexity, making it well-suited for real-time performance evaluation in UAV-based quantum communication systems.

\begin{proposition}
	The conditional PDF of $\mu_q$ given $r_d$ models the probability that a single photon emitted by the transmitter is received successfully by the SPAD detector after propagating through the quantum channel. This probability, conditioned on the beam displacement $r_d$, is expressed as:
	\begin{align} \label{eq:mu_q_pdf_rd}
		&f_{\mu_q \mid r_d}(u) = \exp\left( -\frac{\theta_{\text{FOV}}^2}{2\sigma_{\text{AoA}}^2} \right) \delta(u) \nonumber \\
		&\quad + \left[1 - \exp\left( -\frac{\theta_{\text{FOV}}^2}{2\sigma_{\text{AoA}}^2} \right) \right] \cdot \frac{2(\alpha \beta)^{\frac{\alpha + \beta}{2}}}{c_{\text{pt}} \mu_p(r_d)\,\Gamma(\alpha)\Gamma(\beta)} \nonumber \\
		&\quad \times \left( \frac{u}{c_{\text{pt}} \mu_p(r_d)} \right)^{\frac{\alpha + \beta}{2} - 1}
		K_{\alpha - \beta} \left(2\sqrt{ \alpha \beta \cdot \frac{u}{c_{\text{pt}} \mu_p(r_d)} } \right),
	\end{align}
where $\delta(\cdot)$ denotes the Dirac delta function and $\mu_p(r_d)$ is derived in \eqref{sb1}.
\end{proposition}
\begin{proof}
	Refer to Appendix~\ref{AppB}.
\end{proof}
Proposition~2 characterizes the conditional PDF of the average photon count $\mu_q$ at the receiver, given a beam displacement $r_d$ and conditioned on the transmission of a single photon. It reveals that $\mu_q$ consists of a Dirac delta component at zero, which accounts for the possibility that the photon is lost due to FoV misalignment. This formulation is used to analyze the probability that at least one photon is successfully received at the detector within a given quantum time slot, forming a key component in the performance evaluation of UAV-based quantum communication systems.

According to the transmitter statistics in \eqref{df1}, the number of photons emitted in each quantum slot follows a Poisson distribution. Each emitted photon is subject to deterministic atmospheric absorption \eqref{at1}, random turbulence fading \eqref{at2}, misalignment loss due to pointing errors \eqref{sb1}, geometric acceptance constraints via the FoV indicator \eqref{at3}, and detector inefficiencies.
Since independent thinning of a Poisson process results in another Poisson process, the number of detected photons in each quantum time slot, conditioned on the realization of the random variables $\eta_{\text{turb}}$, $\mu_p$, and $\mu_{\text{FOV}}$, also follows a Poisson distribution:
\begin{align} \label{eq:poisson_cond}
	\mathbb{P}\left(n_q = k \mid \mu_q \right) = \frac{e^{-\mu_q} \mu_q^k}{k!}, \quad k \in \mathbb{N}_0,
\end{align}
where $n_q$ denotes the number of photons detected by the SPAD within a single quantum time slot. Also, $\mu_q$ corresponds to the conditional probability that a single transmitted photon is successfully detected, whose distribution conditioned on the beam displacement $r_d$ is formally derived in Proposition~2.

\begin{proposition}
	The probability of detecting at least one photon in a quantum time slot, denoted by $\mathbb{P}(n_q \geq 1)$, is approximated as:
	\begin{align} \label{eq:laplace_approx_final}
		\mathbb{P}(n_q \geq 1  ) 
		=  \int_0^\infty \mathbb{P}(n_q \geq 1 \mid r_d) \cdot f_{r_d}(r_d) ~dr_d,
	\end{align}
	where the conditional success probability given beam displacement $r_d$ is approximated by:
	\begin{align} \label{eq:laplace_approx_final3}
		&\mathbb{P}(n_q \geq 1 \mid r_d) 
		=  c_{\text{pt}} \left(1 - \exp\left( -\frac{\theta_{\text{FOV}}^2}{2 \sigma_{\text{AoA}}^2} \right) \right) \nonumber \\
		&~~~\times \sum_{i=1}^{N_g} c_i \exp\left( -\frac{2(x_i - r_d)^2}{w_z^2} \right),
	\end{align}
	where $c_i$ are segment weights and $x_i$ are the grid centers defined in Appendix~\ref{AppA}.
\end{proposition}
\begin{proof}
	Refer to Appendix~\ref{AppC}.
\end{proof}
Proposition~3 provides a simplified approximation for the probability of detecting at least one photon in a quantum time slot over a UAV-based quantum communication channel. Although not in closed form, the expression reduces the problem to a single one-dimensional integral over the beam displacement variable $r_d$, making it significantly more tractable than full multidimensional integration. This formulation enables fast numerical evaluation while capturing the joint effects of pointing error, beam spread, and receiver misalignment. Its accuracy is validated in the simulation section through comparison with Monte Carlo results.

\begin{proposition}
	The average raw key generation rate in a UAV-to-ground quantum communication link, denoted by $R_{\text{key}}$, is given by:
	\begin{align} \label{key6}
		R_{\text{key}} = R_q \cdot \mathbb{P}(n_{\text{eff}} = 1),
	\end{align}
	where
	\begin{align} \label{key5}
		&\mathbb{P}(n_{\text{eff}} = 1) =  \mu_b e^{-\mu_b} +
		\left( e^{-\mu_b} - \tfrac{1}{2} \mu_b e^{-\mu_b} \right) \nonumber \\
		&\quad\times c_{\text{pt}} \left(1 - \exp\left( -\frac{\theta_{\text{FOV}}^2}{2 \sigma_{\text{AoA}}^2} \right) \right)
		\nonumber \\
		&\quad \times \int_0^\infty \sum_{i=1}^{N_g} c_i \exp\left( -\frac{2(x_i - r_d)^2}{w_z^2} \right)
		f_{r_d}(r_d)~dr_d.
	\end{align}
\end{proposition}

\begin{proof}
	 Refer to Appendix~\ref{AppD}.
\end{proof}
This result provides a tractable expression for the average raw key generation rate using a single one-dimensional integral, capturing contributions from both successfully received signal photons and, in some cases, inadvertently accepted background photons. While background-induced bits can lead to key generation, they exhibit a $1/2$ probability of being erroneous due to polarization mismatch. Consequently, although narrowing the FoV is an effective strategy to suppress background photon incidence, residual noise remains a performance-limiting factor, especially in highly lossy quantum channels. The impact of such background contributions on the QBER is analyzed in the next proposition.

\begin{proposition}
	The average QBER for the UAV-to-ground quantum communication link is given by:
	\begin{align} \label{key7}
		&\text{QBER} = \frac{1}{2} \times \nonumber \\
		& \resizebox{0.99\linewidth}{!}{$ \frac{
				\mu_b e^{-\mu_b} \cdot \left[1 - 
				\int_0^\infty \mathbb{P}(n_q \geq 1 \mid r_d) \cdot f_{r_d}(r_d)~dr_d \right]
			}{
				\mu_b e^{-\mu_b} +
				\left(e^{-\mu_b} - \tfrac{1}{2} \mu_b e^{-\mu_b} \right)
				\cdot \int_0^\infty \mathbb{P}(n_q \geq 1 \mid r_d) \cdot f_{r_d}(r_d)~dr_d
			}. $}
	\end{align}
\end{proposition}
\begin{proof}
	Refer to Appendix~\ref{AppE}. 
\end{proof}

\begin{table}[t]
	\centering
	\caption{Simulation Parameters}
	\begin{tabular}{ll}
		\hline
		\textbf{Parameter} & \textbf{Value / Range} \\
		\hline
		\multicolumn{2}{l}{\textit{Fixed Parameters}} \\
		\hline
		$L_z$ & 1000\,m (UAV-to-ground distance) \\
		$r_a$ & 15\,cm (Receiver aperture radius) \\
		$\mu_t$ & 0.5 (Mean photons per pulse) \\
		$\eta_{\text{atm}}$ & 0.4 (Atmospheric transmittance) \\
		$\mu_d$ & 0.6 (Detector efficiency) \\
		$T_{qs}$ & $10^{-8}$\,s (Quantum slot duration) \\
		$r_f$ & 5\,$\mu$m (Fiber core radius) \\
		$L_f$ & 15\,cm (Focal length) \\
		$\alpha, \beta$ & 2.1, 1.8 (Turbulence parameters) \\
		$\lambda$ & 1550\,nm (Wavelength) \\
		$\Delta\lambda$ & 1\,nm (Spectral filter bandwidth) \\
		$N_g$ & 10 (Number of grid segments) \\
		\hline
		\multicolumn{2}{l}{\textit{Swept Parameters}} \\
		\hline
		$w_z$ & 5\,cm to 100\,cm (Beam waist at receiver) \\
		$\sigma_{\theta_e}$ & 50\,$\mu$rad to 2\,mrad (FSM pointing error std-dev) \\
		$\sigma_{\text{AoA}}$ & 50\,$\mu$rad to 200\,$\mu$rad (Receiver misalignment std-dev) \\
		$\theta_{\text{FOV}}$ & 5\,$\mu$rad to 200\,$\mu$rad (Receiver angular FoV) \\
		$B_\lambda$ & $10^{-6}$ to $10^{-4}$ W/m$^2$/sr/nm (Spectral radiance) \\
		\hline
	\end{tabular}
	\label{tab:sim_params}
\end{table}

\section{Simulation Results}\label{simulation}
In this section, we investigate the impact of key system and channel parameters on the performance of a UAV-to-ground quantum communication link. In particular, we focus on key parameters such as beam waist at the receiver, pointing errors from the FSM, receiver angular misalignment, FoV, atmospheric turbulence, and background photon levels.

To validate the derived analytical models, we compare them with Monte Carlo simulation results obtained from $10^6$ quantum time slots. In each slot, signal photons are generated based on a Poisson distribution with mean $\mu_t$, while background photons follow a separate Poisson distribution with mean $\mu_b$, depending on the FoV and background spectral radiance. Polarizations of background photons are randomized to reflect real-world noise behavior. Random samples for channel impairments, such as FSM pointing errors, receiver angular deviations, and turbulence-induced fading, are drawn for each time slot. Using the developed receiver model, signal and noise photons are processed to determine quantum bit detection outcomes. Finally, the quantum key generation rate and QBER are computed across all time slots and compared with their analytical counterparts.

The simulation parameters are selected based on widely accepted values in the quantum communication literature, as summarized in Table~\ref{tab:sim_params}. These include key geometric, optical, and statistical parameters that influence photon transmission, beam propagation, receiver coupling, and detection probabilities. To evaluate system robustness and understand the influence of critical design factors, certain parameters, such as beam waist, pointing errors, receiver misalignment, FoV, and background spectral radiance, are swept over a wide range. The exact values used in each scenario are annotated within the corresponding figures throughout this section.

\begin{figure}
	\begin{center}
		\includegraphics[width=3.0 in]{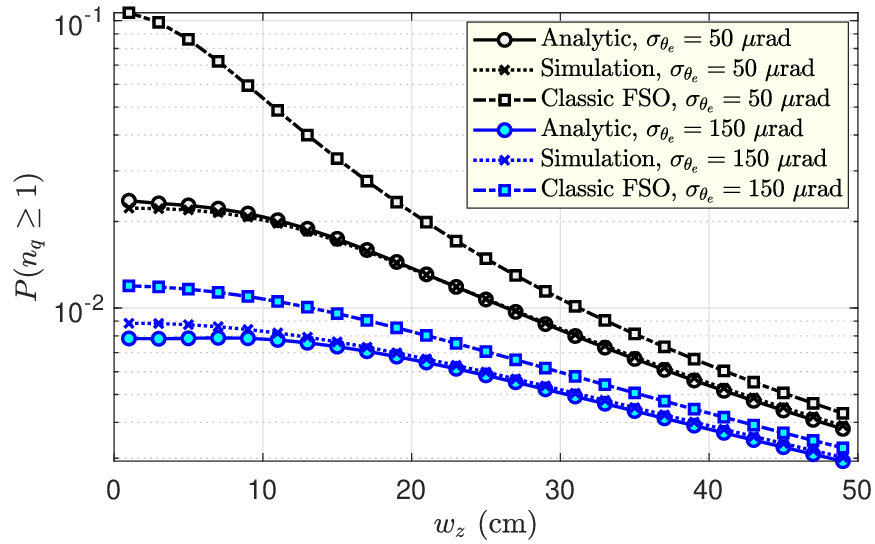}
		\caption{End-to-end single-photon detection probability versus beam waist $w_z$ for two FSM tracking error levels: $\sigma_{\theta_e} = 50~\mu$rad and $150~\mu$rad. Results compare the classical FSO wide-beam approximation in \eqref{xc7}, the proposed analytical model from Proposition~3, and Monte Carlo simulation.}
		\label{ca1}
	\end{center}
\end{figure}

Before analyzing key rate and error performance, it is important to first examine the end-to-end single-photon reception probability, from the transmitter to the SPAD, under varying beam and alignment conditions. Fig.~\ref{ca1} presents this probability as a function of beam waist $w_z$ ranging from 5\,cm to 100\,cm, for two different levels of UAV tracking error: $\sigma_{\theta_e} = 50~\mu$rad and $150~\mu$rad. These values reflect high-precision FSM-based stabilization systems, which are necessary for reliable single-photon detection. As shown, increasing the tracking error from $50~\mu$rad to $150~\mu$rad reduces the detection probability by more than an order of magnitude, dropping below $10^{-2}$.
While classical FSO systems can tolerate misalignments in the milliradian range \cite{dabiri2019tractable,dabiri2019optimal}, quantum communication relies on tightly focused beams and low photon flux, requiring micro-radian scale alignment accuracy. Although this increases system cost, recent advances in compact FSMs make such precision feasible.
The results also highlight that the optimal beam waist lies well below 10\,cm, in contrast to classical FSO systems, where optimal values are typically in the meter range \cite{dabiri2019tractable,dabiri2019optimal}. Therefore, the wide-beam approximation in \eqref{xc7}, based on the assumption $w_z \gg r_a$, becomes invalid in the quantum regime. This approximation is widely adopted in classical FSO link analyses and simulations due to its high accuracy under large beam spot sizes and its analytical simplicity. However, as observed in Fig.~\ref{ca1}, in UAV-based quantum links, where optimal beam waists are tightly focused to maximize single-photon detection, this assumption no longer holds, leading to significant deviations from actual performance. In contrast, the proposed grid-based model in Proposition~3, even with only $N_g=10$ segments, closely follows the simulation results, confirming its accuracy and efficiency, while significantly reducing the computational complexity of the analytical evaluations.

\begin{figure}
	\centering
	\subfloat[] {\includegraphics[width=3.0 in]{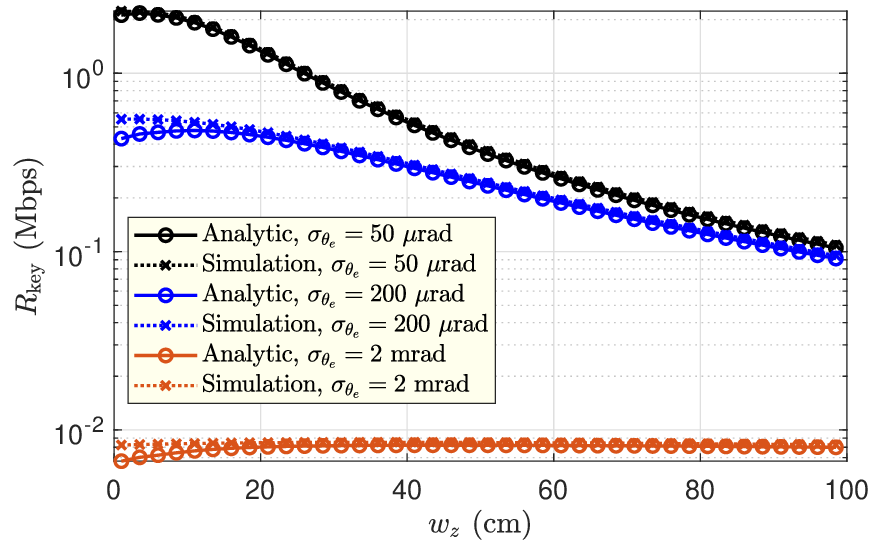}
		\label{cb1}
	}
	\hfill
	\subfloat[] {\includegraphics[width=3.0 in]{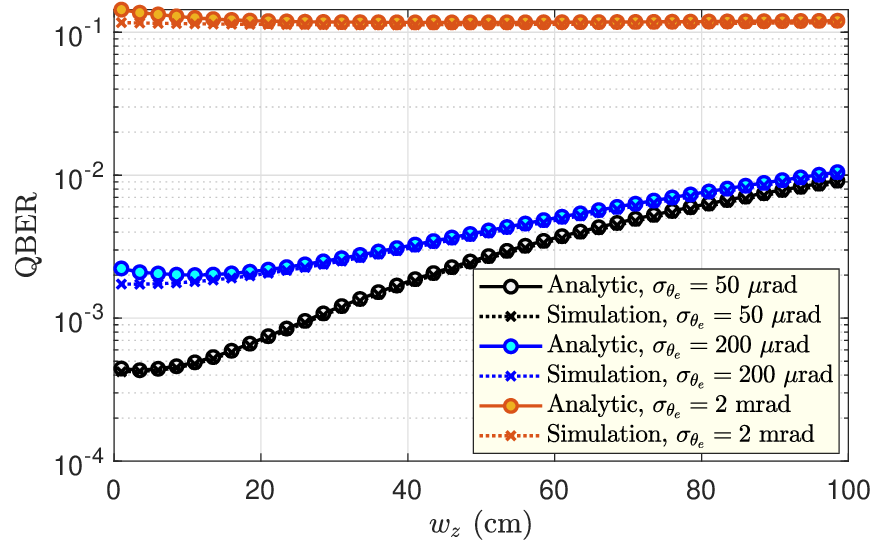}
		\label{cb2}
	}
	\caption{Impact of beam waist $w_z$ on system performance for various UAV tracking errors $\sigma_{\theta_e} = 50~\mu$rad, 200~$\mu$rad, and 2~mrad. Receiver misalignment is fixed at $\sigma_{\text{AoA}} = 50~\mu$rad, and background radiance is $B_\lambda = 10^{-6}$~W/m$^2$/sr/nm. (a) Quantum key generation rate. (b) QBER.}
	\label{cb}
\end{figure}

Fig.~\ref{cb} illustrates the system performance as a function of beam waist $w_z$, evaluated for three different UAV tracking error levels: $\sigma_{\theta_e} = 50~\mu$rad, 200~$\mu$rad, and 2~mrad. In Fig.~\ref{cb1}, the quantum key generation rate is plotted, while Fig.~\ref{cb2} shows the corresponding QBER. The receiver angular misalignment and background spectral radiance are fixed at $\sigma_{\text{AoA}} = 50~\mu$rad and $B_\lambda = 10^{-6}$~W/m$^2$/sr/nm, respectively.
As the shown in the results, with $\sigma_{\theta_e} = 2$~mrad, which represents an alignment level considered acceptable in classical FSO systems, the key rate drops below 100\,kbps, despite the transmitter operating at $T_{qs} = 10^{-8}$\,s, enabling up to 100\,Mbps under ideal conditions. More critically, the QBER exceeds $10^{-1}$, rendering the link unsuitable for secure quantum key distribution.
On the other hand, improving the tracking accuracy to $\sigma_{\theta_e} = 50~\mu$rad significantly enhances performance. For optimized beam waist values, the quantum key rate surpasses 2\,Mbps, while the QBER drops below $10^{-3}$, which is a level compatible with practical error correction and privacy amplification protocols. These results highlight the critical importance of high-precision alignment in UAV-based quantum communication.

\begin{figure}
	\centering
	\subfloat[] {\includegraphics[width=3.0 in]{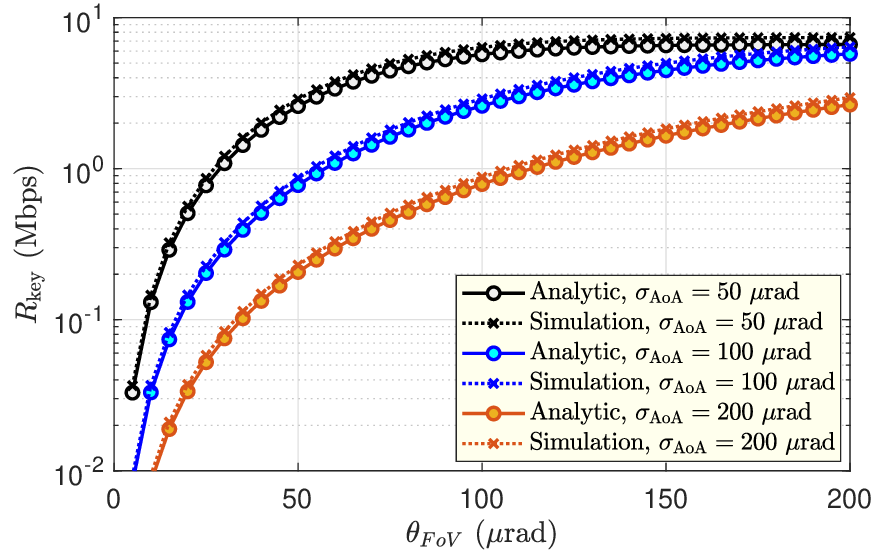}
		\label{cc1}
	}
	\hfill
	\subfloat[] {\includegraphics[width=3.0 in]{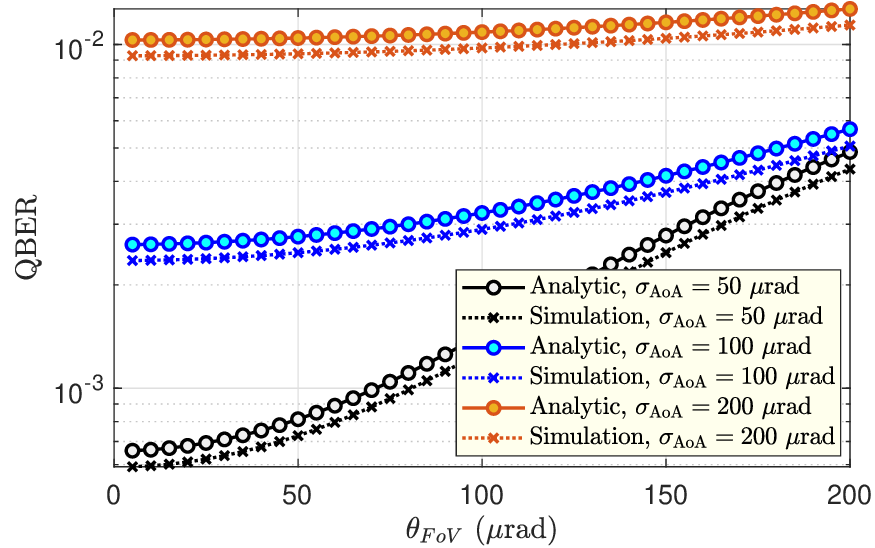}
		\label{cc2}
	}
	\caption{Impact of receiver misalignment $\sigma_{\text{AoA}}$ on system performance with $\sigma_{\theta_e} = 100~\mu$rad and background spectral radiance $B_\lambda = 10^{-6}$\,W/m$^2$/sr/nm. (a) Quantum key generation rate  and (b) QBER versus $\theta_{\text{FoV}}$ for $\sigma_{\text{AoA}} = 50$, 100, and 200\,$\mu$rad.}
	\label{cc}
\end{figure}

Fig.~\ref{cc} illustrates the impact of receiver angular misalignment, modeled by $\sigma_{\text{AoA}}$, on system performance. Figure~\ref{cc1} shows the quantum key generation rate, while Fig. \ref{cc2} presents the QBER for three values of $\sigma_{\text{AoA}} = 50$, 100, and 200\,$\mu$rad. 
Although this parameter is often neglected in classical FSO systems, the results clearly highlight its importance in quantum links. While the receiver is located on the ground and considered stable, continuous angular compensation is required due to UAV movement. This demands a precision stabilizer capable of dynamically aligning the optical axis of the FoV optics—comprising a lens and fiber—to the instantaneous tracking and alignment, which introduces additional weight and cost.
Because most time slots in quantum communication contain no signal-photon due to severe losses, minimizing background photon acceptance is critical. A narrower $\theta_{\text{FoV}}$ suppresses background photon detection, but at the expense of signal loss under misalignment. The simulation confirms that the optimal $\theta_{\text{FoV}}$ must strike a balance: wide enough to tolerate $\sigma_{\text{AoA}}$ but narrow enough to limit QBER.
As observed, increasing $\theta_{\text{FoV}}$ improves the key rate by accepting more single photons, reaching several Mbps for low $\sigma_{\text{AoA}}$. However, it also allows for more background photons, which rapidly increases QBER. Thus, optimal FoV design must be jointly tuned with receiver alignment performance.

\begin{figure}
	\centering
	\subfloat[] {\includegraphics[width=3.0 in]{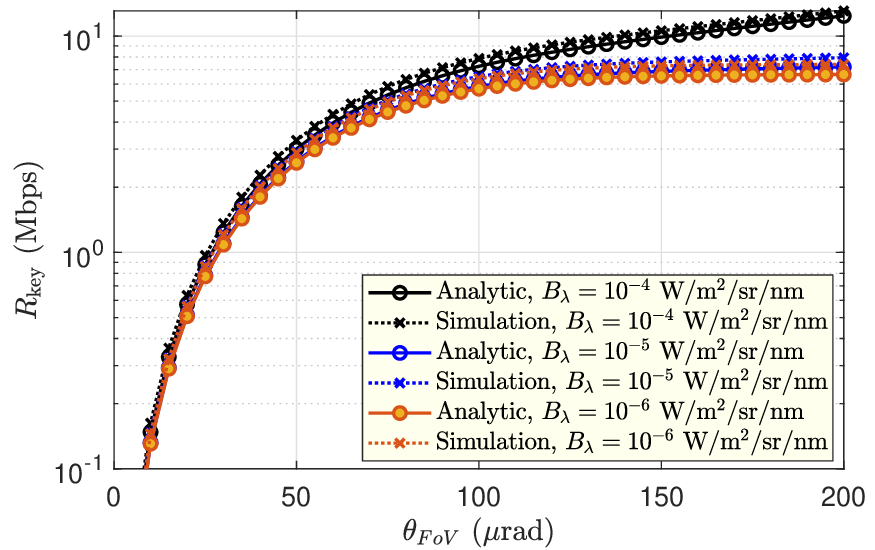}
		\label{cd1}
	}
	\hfill
	\subfloat[] {\includegraphics[width=3.0 in]{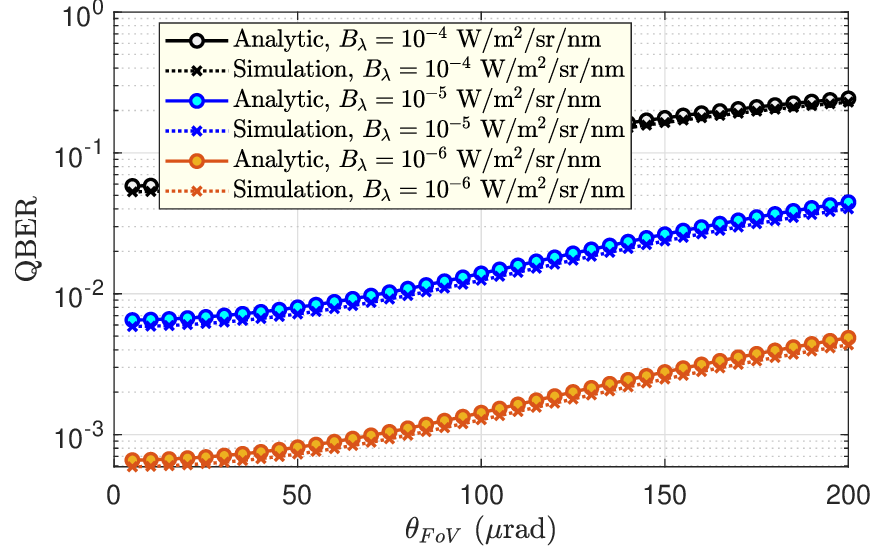}
		\label{cd2}
	}
	\caption{Effect of background radiance $B_\lambda$ on system performance for $\sigma_{\theta_e} = 100~\mu$rad and $\sigma_{\text{AoA}} = 50~\mu$rad. (a) Quantum key generation rate  and (b) QBER versus $\theta_{\text{FoV}}$ for $B_\lambda = 10^{-6}$, $10^{-5}$, and $10^{-4}$\,W/m$^2$/sr/nm.}
	\label{cd}
\end{figure}

In Fig.~\ref{cd}, we further examine how variations in background light intensity, characterized by $B_\lambda$, affect both the optimal $\theta_{\text{FoV}}$ and the overall performance of the system. Subfigures~\ref{cd1} and \ref{cd2} plot the quantum key rate and QBER, respectively, for three representative values of background radiance: $10^{-6}$, $10^{-5}$, and $10^{-4}$\,W/m$^2$/sr/nm. The UAV tracking and receiver angular errors are fixed at $\sigma_{\theta_e} = 100~\mu$rad and $\sigma_{\text{AoA}} = 50~\mu$rad, respectively.
This analysis is particularly relevant for UAV-to-ground links, where the receiver aperture points skyward. Due to the UAV's changing position, the angle of solar incidence and environmental scattering vary over time, causing significant fluctuations in background levels. While values of $B_\lambda$ in the range of $10^{-4}$ are typical for bright daylight, they can drop to $10^{-7}$ or lower at night or under cloudy conditions.
As the results show, high background levels ($B_\lambda = 10^{-4}$) also enable strong key generation rates at large $\theta_{\text{FoV}}$ due to increased signal photon capture. However, this comes at the cost of a very high QBER, rendering the key unusable. To achieve acceptable error rates, the FoV must be significantly reduced to levels that consequently reduce the key rate. Hence, adaptive tuning of $\theta_{\text{FoV}}$ based on ambient background conditions is essential to maintain high key rate and low error rate in real-world deployments.

\section{Conclusion}
This paper developed a comprehensive analytical and simulation framework for modeling UAV-to-ground quantum communication channels. Unlike classical FSO links, quantum channels impose stringent requirements on alignment precision, background suppression, and beam geometry due to their reliance on single-photon transmission. Our contributions include a tractable yet accurate photon capture model based on a grid-based approximation, the derivation of closed-form expressions for key performance metrics, such as quantum key rate and QBER, and a Monte Carlo simulation framework for validating the analytical results.
The results demonstrate that classical assumptions—such as the wide-beam approximation and milliradian-level alignment tolerances—fail in the quantum regime. Specifically, optimal beam waists were found to be in the sub-10\,cm range, where the classical FSO approximation becomes invalid. Moreover, both transmitter-side pointing errors and receiver-side angular misalignments must be controlled at the micro-radian level to ensure viable key generation rates and acceptable QBER levels. We also showed that the receiver’s FoV must be carefully optimized as a trade-off between signal photon collection and background noise rejection, with its optimal value highly sensitive to both mechanical stability and ambient light intensity.
Overall, the proposed framework not only provides critical design insights for practical UAV-based QKD systems but also highlights key differences from classical FSO links that must be addressed to achieve secure and efficient quantum communication in dynamic airborne scenarios.

While this work focused on the UAV-to-ground scenario, which represents one of the most challenging cases due to the transmitter's weight, power, and stability constraints, as well as the high background noise at the upward-facing receiver, the proposed framework is general and can be extended to other configurations, such as ground-to-UAV or terrestrial links, with minor modifications.

Future work may extend this model to incorporate mobility-induced Doppler shifts, wireless synchronization and positioning challenges, network-level multi-hop QKD deployments, and entangled quantum communication scenarios, where two correlated quantum channels must be jointly analyzed.

\appendices

\section{} \label{AppA}
To simplify the computation of $\mu_p$ in \eqref{xc6}, we leverage the rotational symmetry of the problem. Without loss of generality, we assume the photon displacement occurs along the $x$-axis, i.e., $r_{dx} = r_d$ and $r_{dy} = 0$. Under this assumption, the probability becomes:
\begin{align} \label{xc8}
	&\mu_p(r_d) = \frac{2}{\pi w_z^2} \nonumber \\
	&\times \int_{-r_a}^{r_a} \int_{-\sqrt{r_a^2 - x^2}}^{\sqrt{r_a^2 - x^2}} \exp\left( -\frac{2[(x - r_d)^2 + y^2]}{w_z^2} \right) dy\,dx.
\end{align}
The inner integral in \eqref{xc8} can be evaluated in closed form due to the Gaussian form of the integrand with respect to $y$. By isolating the $x$-dependent term and using the standard result for definite Gaussian integrals, the expression simplifies to:
\begin{align} \label{xc9}
	& \mu_p(r_d) = \frac{2}{\sqrt{2\pi} w_z} \int_{-r_a}^{r_a} \exp\left( -\frac{2(x - r_d)^2}{w_z^2} \right) \nonumber \\
	&\times  \text{erf}\left( \sqrt{ \frac{2}{w_z^2} } \sqrt{r_a^2 - x^2} \right) dx,
\end{align}
where $\text{erf}(\cdot)$ is the well-known error function.
This one-dimensional integral accurately models the photon capture probability as a function of beam displacement $r_d$ and spot size $w_z$, and remains valid even when the beam is tightly focused and the large-beam approximation in \eqref{xc7} no longer holds.

Although \eqref{xc9} provides an exact one-dimensional representation of the photon coupling probability, evaluating the integral in closed form is not feasible due to the nonlinear structure of the error function composed with a square root. Nevertheless, the smooth behavior of the integrand allows for an efficient numerical approximation. We divide the interval $[-r_a, r_a]$ into $N_g$ equally spaced segments, each of width $\Delta x = 2r_a / N_g$, and approximate the integrand as locally constant within each subinterval. Letting $x_i$ denote the center of the $i$-th grid interval, we define the contribution from that segment as:
\begin{align}
	\mu_p(r_d, i) = c_i \exp\left( -\frac{2(x_i - r_d)^2}{w_z^2} \right) ,
\end{align}
where
\begin{align}
	c_i = \frac{2\, \Delta x}{\sqrt{2\pi} w_z}  \times \text{erf}\left( \sqrt{ \frac{2}{w_z^2} } \sqrt{r_a^2 - x_i^2} \right).
\end{align}
The total photon coupling probability is then approximated by summing over all subintervals:
\begin{align}
	\mu_p(r_d) \approx \sum_{i=1}^{N_g} \mu_p(r_d, i).
\end{align}
This method provides a simple yet accurate discretization of \eqref{xc9}, and empirical evaluations show that even with a small number of divisions (e.g., $N_g = 10$), the approximation remains highly reliable.

\section{}\label{AppB}
Using \eqref{df1}, \eqref{at1}, \eqref{at2}, \eqref{sb1}, and \eqref{at3}, the end-to-end photon detection probability per slot is given by:
\begin{align} \label{eq:mu_q_final}
	\mu_q =  \underbrace{\mu_t \cdot \eta_{\text{atm}} \cdot \mu_d  \cdot \mu_p(r_d) \cdot \eta_{\text{turb}} }_{\mu_{pt} } \cdot  \mu_{\text{FOV}} .
\end{align}
Given a fixed photon capture probability $\mu_p(r_d)$, the conditional distribution of the effective photon transmission term $\mu_{pt}$ can be derived from the Gamma-Gamma distribution of $\eta_{\text{turb}}$. Letting $c_{\text{pt}} = \mu_t \cdot \eta_{\text{atm}} \cdot \mu_d$. Using \eqref{at2} and \eqref{eq:mu_q_final},
the conditional PDF of $\mu_{pt}$ given $\mu_p(r_d)$ is then obtained via a change of variables from the Gamma-Gamma distribution:
\begin{align} \label{eq:mupt_pdf}
	f_{\mu_{pt} \mid \mu_p}(u) &= \frac{2(\alpha \beta)^{\frac{\alpha + \beta}{2}}}{c_{\text{pt}}  \mu_p(r_d) \,\Gamma(\alpha)\Gamma(\beta)} \left( \frac{u}{c_{\text{pt}}  \mu_p(r_d)} \right)^{\frac{\alpha + \beta}{2} - 1}    \nonumber \\
	&~~~\times K_{\alpha - \beta} \left(2\sqrt{ \alpha \beta \cdot \frac{u}{c_{\text{pt}}  \mu_p(r_d) } } \right),
\end{align}
for $u > 0$.

Using the AoA distribution model in \eqref{aoa1} and the binary visibility indicator definition in \eqref{at3}, the random variable $\mu_{\text{FOV}}$ is a Bernoulli-type variable that takes the value 1 if the angle of arrival lies within the receiver's acceptance cone, and 0 otherwise. Since the norm $\|\boldsymbol{\theta}_{\text{AoA}}\|$ follows a Rayleigh distribution with scale parameter $\sigma_{\text{AoA}}$, the distribution of $\mu_{\text{FOV}}$ is given by:
\begin{align} \label{aoa2}
	\begin{cases}
		\mathbb{P}(\mu_{\text{FOV}} = 1) &= \mathbb{P}\left( \|\boldsymbol{\theta}_{\text{AoA}}\| \leq \theta_{\text{FOV}} \right) = 1 - \exp\left( -\frac{\theta_{\text{FOV}}^2}{2\sigma_{\text{AoA}}^2} \right), \\
		\mathbb{P}(\mu_{\text{FOV}} = 0) &= \exp\left( -\frac{\theta_{\text{FOV}}^2}{2\sigma_{\text{AoA}}^2} \right).
	\end{cases}
\end{align}
Hence, the PDF of $\mu_{\text{FOV}}$ is expressed in the impulse (Dirac delta) form as:
\begin{align} \label{eq:fov_pdf}
	f_{\mu_{\text{FOV}}}(u) &= \exp\left( -\frac{\theta_{\text{FOV}}^2}{2\sigma_{\text{AoA}}^2} \right) \delta(u) \nonumber \\
	&~~~+ \left[1 - \exp\left( -\frac{\theta_{\text{FOV}}^2}{2\sigma_{\text{AoA}}^2} \right) \right] \delta(u - 1),
\end{align}
where $\delta(\cdot)$ denotes the Dirac delta function.
Using the decomposition of the mean photon count in \eqref{eq:mu_q_final} as the product $\mu_q = \mu_{pt} \cdot \mu_{\text{FOV}}$, and combining the conditional Gamma-Gamma distribution of $\mu_{pt}$ with the binary distribution of $\mu_{\text{FOV}}$, the conditional PDF of $\mu_q$ given $\mu_p(r_d)$ is obtained in \eqref{eq:mu_q_pdf_rd}.

\section{}\label{AppC}
Based on the conditional Poisson model in \eqref{eq:poisson_cond}, the probability of detecting at least one photon in a quantum time slot, given a fixed mean photon count $\mu_q$, is expressed as:
\begin{align} \label{eq:poisson_cond_prob}
	\mathbb{P}(n_q \geq 1 \mid \mu_q) = 1 - e^{-\mu_q}.
\end{align}
Using this expression and the conditional distribution of $\mu_q$ derived in \eqref{eq:mu_q_pdf_rd}, the conditional probability of photon detection given a fixed photon capture probability $\mu_p$ can be obtained by averaging over $\mu_q$:
\begin{align} \label{eq:poisson_cond_success_final}
	\mathbb{P}(n_q \geq 1 \mid \mu_p) &= \int_0^\infty \left(1 - e^{-u} \right) f_{\mu_q \mid \mu_p}(u) \, du.
\end{align}
While the above expression is conditioned on a fixed photon capture probability $\mu_p$, in practice $\mu_p$ is a random variable that depends on the beam displacement $r_d$. The displacement $r_d$ follows a Rayleigh distribution due to the Gaussian FSM angular jitter described in \eqref{xc2}, with corresponding PDF:
\begin{align} \label{eq:rd_pdf}
	f_{r_d}(r) = \frac{r}{\sigma_{r_d}^2} \exp\left(-\frac{r^2}{2\sigma_{r_d}^2} \right), \quad r \geq 0,
\end{align}
where $\sigma_{r_d}^2 = \sigma_{\theta_e}^2 L_z^2$ is the variance of lateral beam displacement.

To account for the randomness in photon capture due to lateral beam displacement, we substitute the conditional mean photon count $\mu_q$ in \eqref{eq:poisson_cond_prob} with its distribution conditioned on $r_d$, as given by \eqref{eq:mu_q_pdf_rd}. This yields the conditional probability of detecting at least one photon given the beam displacement $r_d$:
\begin{align} \label{eq:poisson_success_cond_rd_full}
	&\mathbb{P}(n_q \geq 1 \mid r_d) 
	= \left(1 - \exp\left( -\frac{\theta_{\text{FOV}}^2}{2\sigma_{\text{AoA}}^2} \right) \right) \nonumber \\
	&\quad \times \int_0^\infty \left(1 - e^{-u} \right)
	\cdot \frac{2(\alpha \beta)^{\frac{\alpha + \beta}{2}}}{c_{\text{pt}} \mu_p(r_d)\,\Gamma(\alpha)\Gamma(\beta)} 
	\left( \frac{u}{c_{\text{pt}} \mu_p(r_d)} \right)^{\frac{\alpha + \beta}{2} - 1} \nonumber \\
	&\quad \times K_{\alpha - \beta} \left(2\sqrt{ \alpha \beta \cdot \frac{u}{c_{\text{pt}} \mu_p(r_d)} } \right) du,
\end{align}
where the dependence of $\mu_p$ on $r_d$ is given by the approximation in \eqref{sb1}, and $f_{r_d}(r)$ is defined in \eqref{eq:rd_pdf}.

By applying the change of variable $z = u / (c_{\text{pt}} \mu_p)$ to the exponential integral in the Laplace domain and approximating the exponential term as $e^{-z c_{\text{pt}} \mu_p} \approx 1 - z c_{\text{pt}} \mu_p$ for small arguments, the resulting integral becomes amenable to analytical evaluation. Using the identity from Gradshteyn and Ryzhik~\cite[Eq.~6.621.3]{gradshteyn2007table},
\begin{align}
	\int_0^\infty z^{\rho - 1} K_\nu(2 a \sqrt{z})\, dz = \frac{1}{2 a^{2\rho}} \Gamma(\rho + \nu) \Gamma(\rho - \nu),
\end{align}
with the parameter substitutions $\rho = \frac{\alpha + \beta}{2} + 1$, $\nu = \alpha - \beta$, and $a = \sqrt{\alpha \beta}$, we derive a simplified approximation for the exponential moment of the Gamma-Gamma distributed transmittance.
As a result, the overall probability of detecting at least one photon—averaged over the Rayleigh-distributed beam displacement $r_d$ from \eqref{eq:rd_pdf}—is approximated by the single integral expression given in~\eqref{eq:laplace_approx_final}.

\section{} \label{AppD}
In this appendix, we derive the average quantum key generation rate for a UAV-to-ground link, accounting for both signal and background photons. A raw key bit is generated whenever exactly one photon is detected—regardless of its source. Since the receiver cannot distinguish between signal and background photons, any photon arriving within the FoV may contribute to key generation, provided no other photon is detected in the same time slot.
Notably, even background photons can produce valid key bits if no signal photon is present, or if the signal photon is misaligned and the background photon reaches the corresponding SPAD. Given their random polarization, background photons match the correct detection basis with probability $1/2$.

Let $n_q$ and $n_b$ denote the number of signal and background photons detected in a given quantum slot, respectively. The total number of detected photons in the slot is:
\begin{align}
	n_{\text{tot}} = n_q + n_b.
\end{align}

A time slot contributes to the raw key if and only if exactly one photon is effectively detected by the SPADs, such that a valid bit can be extracted. This occurs under the following disjoint events:
\begin{align} \label{key1}
	\mathbb{P}(n_{\text{eff}} = 1) &= 
	\underbrace{\mathbb{P}(n_q \geq 1, n_b = 0)}_{\text{State 1: signal only}} 
	+ \underbrace{\mathbb{P}(n_q = 0, n_b = 1)}_{\text{State 2: background only}} \nonumber \\
	&\quad + \underbrace{\tfrac{1}{2} \cdot \mathbb{P}(n_q \geq 1, n_b = 1)}_{\text{State 3: signal and background, non-interfering}}.
\end{align}
Here, the factor $\tfrac{1}{2}$ in \textit{State 3} accounts for the fact that a background photon reaches the same SPAD as the signal photon with probability $1/2$, due to its random polarization.

Using \eqref{key1}, the first case, \textit{State 1}, corresponds to the successful detection of at least one signal photon and no background photon. This yields a clean detection scenario with:
\begin{align} \label{key2}
	&\mathbb{P}_{\text{S1}} = e^{-\mu_b} \cdot \int_0^\infty \Bigg[
	c_{\text{pt}} \left(1 - \exp\left( -\frac{\theta_{\text{FOV}}^2}{2 \sigma_{\text{AoA}}^2} \right) \right) \nonumber \\
	&\times \sum_{i=1}^{N_g} c_i \exp\left( -\frac{2(x_i - r_d)^2}{w_z^2} \right)
	\Bigg] \cdot f_{r_d}(r_d)~dr_d.
\end{align}
In \textit{State 2}, the signal photon is completely absent, and a single background photon is detected. Since it is the only photon present in the time slot, it generates a raw key bit:
\begin{align} \label{key3}
	&\mathbb{P}_{\text{S2}} = \mu_b e^{-\mu_b} \cdot \Bigg[1 - \int_0^\infty \Bigg(
	c_{\text{pt}} \left(1 - \exp\left( -\frac{\theta_{\text{FOV}}^2}{2 \sigma_{\text{AoA}}^2} \right) \right) \nonumber \\
	&\times \sum_{i=1}^{N_g} c_i \exp\left( -\frac{2(x_i - r_d)^2}{w_z^2} \right)
	\Bigg) \cdot f_{r_d}(r_d)~dr_d \Bigg].
\end{align}
The third case, \textit{State 3}, considers the simultaneous detection of signal and background photons. A valid key bit is only registered if the background photon does not interfere with the SPAD corresponding to the signal.  It generates a raw key bit:
\begin{align} \label{key4}
	&\mathbb{P}_{\text{S3}} = \tfrac{1}{2} \cdot \mu_b e^{-\mu_b} \cdot \int_0^\infty \Bigg[
	c_{\text{pt}} \left(1 - \exp\left( -\frac{\theta_{\text{FOV}}^2}{2 \sigma_{\text{AoA}}^2} \right) \right) \nonumber \\
	&\times \sum_{i=1}^{N_g} c_i \exp\left( -\frac{2(x_i - r_d)^2}{w_z^2} \right)
	\Bigg] \cdot f_{r_d}(r_d)~dr_d.
\end{align}
Based on \eqref{key2}--\eqref{key4}, the total probability of accepting a quantum time slot for raw key generation is derived in \eqref{key5}.

Accordingly, the average raw key generation rate becomes:
\begin{align}
	R_{\text{key}} = R_q \cdot \mathbb{P}(n_{\text{eff}} = 1),
\end{align}
where $R_q = 1/T_{qs}$ denotes the quantum transmission rate.

\section{} \label{AppE}
In this appendix, we derive the average QBER for the UAV-to-ground link by examining the composition of accepted raw key bits. As discussed in \eqref{key1}, valid key bits are generated in three disjoint scenarios, among which only one contributes to bit errors.
Specifically, \textit{State 2} corresponds to the case where no signal photon is present and a single background photon is detected. Since background photons carry no information about the transmitted bit and their polarization is random, any key bit generated in this state has a $50\%$ chance of being erroneous.
In contrast, bits generated in \textit{State 1} (signal-only) and \textit{State 3} (signal with non-interfering background) are assumed to be error-free, assuming ideal polarization alignment and perfect detector performance \cite{zhang2012characteristics}.
Therefore, based on \eqref{key1}--\eqref{key4}, the average QBER is given by:
\begin{align}
	\text{QBER} = \frac{1}{2} \cdot \frac{\mathbb{P}_{\text{S2}}}{\mathbb{P}(n_{\text{eff}} = 1)},
\end{align}
Therefore, based on \eqref{key1}--\eqref{key4}, the average QBER is obtained \eqref{key7}.


\balance


\begin{thebibliography}{10}
	\providecommand{\url}[1]{#1}
	\csname url@samestyle\endcsname
	\providecommand{\newblock}{\relax}
	\providecommand{\bibinfo}[2]{#2}
	\providecommand{\BIBentrySTDinterwordspacing}{\spaceskip=0pt\relax}
	\providecommand{\BIBentryALTinterwordstretchfactor}{4}
	\providecommand{\BIBentryALTinterwordspacing}{\spaceskip=\fontdimen2\font plus
		\BIBentryALTinterwordstretchfactor\fontdimen3\font minus
		\fontdimen4\font\relax}
	\providecommand{\BIBforeignlanguage}[2]{{%
			\expandafter\ifx\csname l@#1\endcsname\relax
			\typeout{** WARNING: IEEEtran.bst: No hyphenation pattern has been}%
			\typeout{** loaded for the language `#1'. Using the pattern for}%
			\typeout{** the default language instead.}%
			\else
			\language=\csname l@#1\endcsname
			\fi
			#2}}
	\providecommand{\BIBdecl}{\relax}
	\BIBdecl
	
	\bibitem{sood2024cryptography}
	N.~Sood, ``{Cryptography in post Quantum computing era},'' \emph{Available at
		SSRN 4705470}, 2024.
	
	\bibitem{tom2023quantum}
	J.~J. Tom, N.~P. Anebo, B.~A. Onyekwelu, A.~Wilfred, and R.~Eyo, ``Quantum
	computers and algorithms: a threat to classical cryptographic systems,''
	\emph{Int. J. Eng. Adv. Technol}, vol.~12, no.~5, pp. 25--38, 2023.
	
	\bibitem{nadlinger2022experimental}
	D.~P. Nadlinger, P.~Drmota, B.~C. Nichol, G.~Araneda, D.~Main, R.~Srinivas,
	D.~M. Lucas, C.~J. Ballance, K.~Ivanov, E.-Z. Tan \emph{et~al.},
	``{Experimental quantum key distribution certified by Bell's theorem},''
	\emph{Nature}, vol. 607, no. 7920, pp. 682--686, 2022.
	
	\bibitem{kong2024secret}
	P.-Y. Kong, ``Secret key rate over multiple relays in quantum key distribution
	for cyber--physical systems,'' \emph{IEEE Transactions on Industrial
		Informatics}, 2024.
	
	\bibitem{kong2022unmanned}
	P.-Y. Kong and Y.~Wang, ``Unmanned aerial vehicle as encryption key distributor
	for secure communications in smart grid,'' \emph{IEEE Internet of Things
		Journal}, vol.~10, no.~8, pp. 6849--6858, 2022.
	
	\bibitem{quintana2019low}
	C.~Quintana, P.~Sibson, G.~Erry, Y.~Thueux, E.~Kingston, T.~Ismail,
	G.~Faulkner, J.~Kennard, K.~Gebremicael, C.~Clark \emph{et~al.}, ``Low size,
	weight and power quantum key distribution system for small form unmanned
	aerial vehicles,'' in \emph{Free-Space Laser Communications XXXI}, vol.
	10910.\hskip 1em plus 0.5em minus 0.4em\relax SPIE, 2019, pp. 240--246.
	
	\bibitem{trinh2024quantum}
	P.~V. Trinh and S.~Sugiura, ``Quantum internet in the sky: Vision, challenges,
	solutions, and future directions,'' \emph{IEEE Communications Magazine},
	vol.~62, no.~10, pp. 62--68, 2024.
	
	\bibitem{trinh2025towards}
	P.~V. Trinh, S.~Sugiura, C.~Xu, and L.~Hanzo, ``{Towards Quantum SAGINs
		Harnessing Optical RISs: Applications, Advances, and the Road Ahead},''
	\emph{IEEE Network}, 2025.
	
	\bibitem{ralegankar2021quantum}
	V.~K. Ralegankar, J.~Bagul, B.~Thakkar, R.~Gupta, S.~Tanwar, G.~Sharma, and
	I.~E. Davidson, ``{Quantum cryptography-as-a-service for secure UAV
		communication: applications, challenges, and case study},'' \emph{Ieee
		Access}, vol.~10, pp. 1475--1492, 2021.
	
	\bibitem{kumar2021survey}
	A.~Kumar, S.~Bhatia, K.~Kaushik, S.~M. Gandhi, S.~G. Devi, D.~A. D.~J. Pacheco,
	and A.~Mashat, ``Survey of promising technologies for quantum drones and
	networks,'' \emph{Ieee Access}, vol.~9, pp. 125\,868--125\,911, 2021.
	
	\bibitem{gao2023qkd}
	Z.~Gao, W.~Fan, and R.~Luo, ``{QKD-based Secure Communication for UAV},'' in
	\emph{2023 24st Asia-Pacific Network Operations and Management Symposium
		(APNOMS)}.\hskip 1em plus 0.5em minus 0.4em\relax IEEE, 2023, pp. 107--112.
	
	\bibitem{kondamuri2024quantum}
	S.~R. Kondamuri, N.~K. Kundu, and A.~Ghanshyala, ``Quantum communication in
	weakly turbulent channels: Error analysis and capacity bounds,'' \emph{IEEE
		Wireless Communications Letters}, 2024.
	
	\bibitem{waghmare2023performance}
	C.~Waghmare, A.~Kothari, and P.~K. Sharma, ``Performance analysis of classical
	data transmission over a quantum channel in the presence of atmospheric
	turbulence,'' \emph{IEEE Communications Letters}, vol.~27, no.~8, pp.
	2127--2131, 2023.
	
	\bibitem{nguyen2024blind}
	C.~T. Nguyen, H.~D. Le, V.~V. Mai, P.~V. Trinh, and A.~T. Pham, ``Blind
	reconciliation with protograph ldpc code extension for fso-based satellite
	qkd systems,'' in \emph{2024 14th International Symposium on Communication
		Systems, Networks and Digital Signal Processing (CSNDSP)}.\hskip 1em plus
	0.5em minus 0.4em\relax IEEE, 2024, pp. 17--22.
	
	\bibitem{trinh2025optical}
	P.~V. Trinh, S.~Sugiura, C.~Xu, and L.~Hanzo, ``{Optical RISs improve the
		secret key rate of free-space QKD in HAP-to-UAV scenarios},'' \emph{IEEE
		Journal on Selected Areas in Communications}, 2025.
	
	\bibitem{vazquez2023quantum}
	A.~V{\'a}zquez-Castro and B.~Samandarov, ``Quantum advantage of binary discrete
	modulations for space channels,'' \emph{IEEE Wireless Communications
		Letters}, vol.~12, no.~5, pp. 903--906, 2023.
	
	\bibitem{chakraborty2025hybrid}
	M.~Chakraborty, A.~Mukherjee, I.~Krikidis, A.~Nag, and S.~Chandra, ``{A Hybrid
		Noise Approach to Modelling of Free-Space Satellite Quantum Communication
		Channel for Continuous-Variable QKD},'' \emph{IEEE Transactions on Green
		Communications and Networking}, 2025.
	
	\bibitem{dabiri2025impact}
	M.~T. Dabiri, M.~Hasna, S.~Althunibat \emph{et~al.}, ``On the impact of
	tracking inaccuracy in space-based quantum key distribution: A stochastic
	geometric approach,'' TechRxiv preprint, May 2025.
	
	\bibitem{alshaer2021reliability}
	N.~Alshaer, A.~Moawad, and T.~Ismail, ``{Reliability and security analysis of
		an entanglement-based QKD protocol in a dynamic ground-to-UAV FSO
		communications system},'' \emph{IEEE Access}, vol.~9, pp. 168\,052--168\,067,
	2021.
	
	\bibitem{alshaer2022performance}
	N.~Alshaer and T.~Ismail, ``{Performance evaluation and security analysis of
		UAV-based FSO/CV-QKD system employing DP-QPSK/CD},'' \emph{IEEE Photonics
		Journal}, vol.~14, no.~3, pp. 1--11, 2022.
	
	\bibitem{kong2024uav}
	P.-Y. Kong, ``{UAV-Assisted Quantum Key Distribution for Secure Communications
		with Resource Limited Devices},'' \emph{IEEE Transactions on Vehicular
		Technology}, 2024.
	
	\bibitem{xue2021airborne}
	Y.~Xue, W.~Chen, S.~Wang, Z.~Yin, L.~Shi, and Z.~Han, ``Airborne quantum key
	distribution: a review,'' \emph{Chinese Optics Letters}, vol.~19, no.~12, p.
	122702, 2021.
	
	\bibitem{al2025unified}
	Y.~H. Al-Badarneh, O.~S. Badarneh, M.~K. Alshawaqfeh, T.~M. Khattab, and M.~O.
	Hasna, ``A unified {MGF}-based performance analysis of quantum communications
	over turbulence channels with pointing errors,'' \emph{IEEE Wireless
		Communications Letters}, 2025.
	
	\bibitem{farid2007outage}
	A.~A. Farid and S.~Hranilovic, ``Outage capacity optimization for free-space
	optical links with pointing errors,'' \emph{Journal of Lightwave technology},
	vol.~25, no.~7, pp. 1702--1710, 2007.
	
	\bibitem{dabiri2019tractable}
	M.~T. Dabiri, S.~M.~S. Sadough, and I.~S. Ansari, ``{Tractable optical channel
		modeling between UAVs},'' \emph{IEEE Transactions on Vehicular Technology},
	vol.~68, no.~12, pp. 11\,543--11\,550, 2019.
	
	\bibitem{dabiri2019optimal}
	M.~T. Dabiri and S.~M.~S. Sadough, ``{Optimal placement of UAV-assisted
		free-space optical communication systems with DF relaying},'' \emph{IEEE
		Communications Letters}, vol.~24, no.~1, pp. 155--158, 2019.
	
	\bibitem{foletto2022security}
	G.~Foletto, F.~Picciariello, C.~Agnesi, P.~Villoresi, and G.~Vallone,
	``Security bounds for decoy-state quantum key distribution with arbitrary
	photon-number statistics,'' \emph{Physical Review A}, vol. 105, no.~1, p.
	012603, 2022.
	
	\bibitem{pampaloni2004gaussian}
	F.~Pampaloni and J.~Enderlein, ``{Gaussian, hermite-gaussian, and
		laguerre-gaussian beams: A primer},'' \emph{arXiv preprint physics/0410021},
	2004.
	
	\bibitem{ghassemlooy2019optical}
	\emph{Optical wireless communications: system and channel modelling with
		Matlab{\textregistered}}.\hskip 1em plus 0.5em minus 0.4em\relax CRC press,
	2019.
	
	\bibitem{dabiri2025novel}
	M.~T. Dabiri and M.~Hasna, ``{A Novel MRR-UAV-Based Relay With Optical Network
		Coding: A Comparative Study With Optical IRS and Conventional UAV
		Relaying},'' \emph{IEEE Journal on Selected Areas in Communications},
	vol.~43, no.~5, pp. 1607--1620, 2025.
	
	\bibitem{zhang2012characteristics}
	H.~Zhang, H.~Yin, H.~Jia, S.~Chang, and J.~Yang, ``Characteristics of
	non-line-of-sight polarization ultraviolet communication channels,''
	\emph{Applied optics}, vol.~51, no.~35, pp. 8366--8372, 2012.
	
	\bibitem{andrews2005laser}
	L.~C. Andrews and R.~L. Phillips, \emph{Laser Beam Propagation through Random
		Media}.\hskip 1em plus 0.5em minus 0.4em\relax SPIE Press, 2005.
	
	\bibitem{kim2022photon}
	J.-W. Kim, J.-S. Cho, C.~Sacarelo, N.~D.~F. Fitri, J.-S. Hwang, and J.-K.~K.
	Rhee, ``Photon-counting statistics-based support vector machine with
	multi-mode photon illumination for quantum imaging,'' \emph{Scientific
		Reports}, vol.~12, no.~1, p. 16594, 2022.
	
	\bibitem{dabiri2018channel}
	M.~T. Dabiri, S.~M.~S. Sadough, and M.~A. Khalighi, ``{Channel modeling and
		parameter optimization for hovering UAV-based free-space optical links},''
	\emph{IEEE Journal on Selected Areas in Communications}, vol.~36, no.~9, pp.
	2104--2113, 2018.
	
	\bibitem{dabiri2024modulating}
	M.~T. Dabiri, M.~Hasna, S.~Althunibat, and K.~Qaraqe, ``Modulating
	retroreflector-based satellite-to-ground optical communications: Acquisition,
	sensing and positioning,'' \emph{IEEE Transactions on Communications}, 2024.
	
	\bibitem{gradshteyn2007table}
	I.~S. Gradshteyn and I.~M. Ryzhik, \emph{Table of Integrals, Series, and
		Products}, 7th~ed., A.~Jeffrey and D.~Zwillinger, Eds.\hskip 1em plus 0.5em
	minus 0.4em\relax Academic Press, 2007.
	
\end{thebibliography}
\end{document}